\documentclass[conference]{IEEEtran}
\usepackage[cmex10]{amsmath}
\usepackage{amsmath}
\usepackage{mdwmath}
\usepackage{mdwtab}
\usepackage{verbatim}
\usepackage{setspace}
\usepackage{url}
\usepackage{amssymb}
\usepackage{graphicx}
\usepackage{amsthm}
\usepackage{latexsym}
\usepackage{amsfonts}
\usepackage{amscd}
\usepackage{young}
\usepackage[vcentermath]{youngtab}
\usepackage{xcolor}

\theoremstyle{plain}
\newtheorem{thm}{Theorem}

\newtheorem*{example}{Example}
\newtheorem{theorem}[thm]{Theorem}
\newtheorem{proposition}[thm]{Proposition}

\newtheorem{lemma}[thm]{Lemma}

\newtheorem*{definition}{Definition}

\newtheorem*{remark}{Remark}

\begin{document}
\title{Multipermutation Ulam Sphere Analysis 
Toward Characterizing Maximal Code Size}

\author
{\IEEEauthorblockN
{Justin Kong and Manabu Hagiwara}
\IEEEauthorblockA{
Department of Mathematics and Informatics, \\
Graduate School of Science,
Chiba University\\
Email: jkong@math.s.chiba-u.ac.jp, hagiwara@math.s.chiba-u.ac.jp
}
}

\maketitle

\begin{abstract}
\boldmath
Permutation codes, in the form of \textit{rank modulation},
 have shown promise for applications such as flash memory.  
 One of the metrics recently suggested as appropriate for rank 
modulation is the Ulam metric, which measures the minimum 
translocation distance between permutations.  Multipermutation 
codes have also been proposed as a generalization of 
permutation codes that would improve code size (and 
consequently the code rate).  In this paper we analyze the 
Ulam metric in the context of multipermutations, noting some 
similarities and differences between the Ulam metric in the 
context of permutations.  We also consider sphere sizes for 
multipermutations under the Ulam metric and resulting bounds 
on code size.
\end{abstract}

\IEEEpeerreviewmaketitle

\section{Introduction}
Permutation (and multipermutation) codes 
were invented as early as 
the 1960's, when Slepian proposed constructing a code by permuting the 
order of the numbers of an initial sequence \cite{Slepian}.  More recently, 
Jiang et al. proposed permutation codes utilizing the Kendall-$\tau$ metric 
for use in flash memory via the \textit{rank modulation} scheme \cite{Jiang}. 
Since then, permutation codes 
and their generalization to multipermutation codes
 have been a hot topic in the research community with various 
related schemes being suggested 
\cite{Barg, Coding with permutations, Farnoud, 
Farnoudtwo, Jiangtwo, Lim}.

One scheme of particular interest was 
the proposal of Farnoud et al. to utilize the Ulam metric in place of the Kendall-$\tau$ metric \cite{Farnoud} 
and subsequent study 
expounded upon code size bounds \cite{Gologlu}.
The Ulam metric measures the minimum number of translocations
needed to transform one permutation into another, whereas 
the Kendall-$\tau$ metric measures the minimum 
adjacent transpositions needed to transform one permutation into another.  
Errors in flash memory devices occur when cell charges leak 
or when rewriting operations cause overshoot errors resulting 
in inaccurate charge levels.   
While the Kendall-$\tau$ metric is suitable for correcting 
relatively small errors of this nature, the 
Ulam metric would be more robust to large 
charge leakages or overshoot errors within a cell.  

However, there is a trade-off in code size when 
rank modulation is used in conjunction with 
the Ulam metric instead of the Kendall-$\tau$ metric. 
The Ulam distance between permutations is 
always less than or equal to the Kendall-$\tau$ 
distance between permutations, which implies that 
the maximum code size for a permutation code utilizing 
the Ulam metric is less than or equal to the maximum 
code size of a permutation code utilizing the 
Kendall-$\tau$ metric \cite{Farnoud}.
One possible compensation for this trade-off is the 
generalization from permutation codes to multipermutation 
codes, which improves the maximum possible code size \cite{Farnoudtwo}.  

In flash memory devices, permutations or multipermutations 
may be modeled physically by relative rankings of cell charges. 
The number of possible messages is limited by the number 
of distinguishable relative rankings.  However, 
it was shown in \cite{Farnoudtwo} that multipermutations 
may significantly increase the total possible messages 
compared to ordinary permutations. 
For example, if only $k$ different charge levels are possible, 
then permutations of length $k$ can be stored.  
Hence, in $r$ blocks of length $k,$ one may store $(k!)^r$ 
potential messages.  On the other hand, if one uses $r$-regular 
multipermutations in the same set of blocks, then 
$(kr)!/(r!)^k$ potential messages are possible. 

Bounds on permutation codes in the Ulam metric were studied in 
\cite{Farnoud} and \cite{Gologlu}.
In \cite{self}, the nonexistence of nontrivial perfect permutation 
codes in the Ulam metric was proven by examining the size of 
Ulam spheres, spheres comprised of all permutations within 
a given Ulam distance of a particular permutation. 
However, no similar study of Multipermutation Ulam spheres 
exists, and currently known bounds on code size 
do not always consider the problem of differing sphere sizes. 
The current paper examines Ulam sphere sizes in the context of 
multipermutations and provides new bounds on code size.

The paper is organized as follows: 
First, Section \ref{prelims} defines notation and basic concepts 
used in the paper. Next, Section \ref{multiperm} compares 
properties of the Ulam metric as defined for permutations 
and multipermutations, and then provides a simplification of 
the $r$-regular Ulam metric for multipermutations 
(Lemmas \ref{n-l} and \ref{translocations}). 
Section \ref{young} considers an application of 
Young Tableaux and the RSK-correspondence 
to calculate $r$-regular Ulam sphere sizes
(Lemma \ref{flambda applied} and Prop. \ref{onesphere}). 
Section \ref{upper} then discusses duplicate translocation 
sets and a method of calculating the size of 
spheres of radius $t=1$ for any center (Thm. \ref{ballcalc}). 
Section \ref{minmax} follows, demonstrating minimal and 
maximal sphere sizes 
(Lemmas \ref{esphere} and \ref{omegasphere}) 
and providing 
both lower and upper bounds on code size
(Lemmas \ref{upperbound}, 
\ref{perfect bound}, and \ref{G-V bound}). 
Finally Section \ref{conclusion} gives some concluding remarks.

\section{Preliminaries and Notation}\label{prelims}
In this section we introduce notation and definitions 
used in this paper.  Unless otherwise stated, definitions are 
based on conventions established in 
\cite{Farnoud}, \cite{Farnoudtwo}, and \cite{self}.
Throughout this paper $n$ and $r$ are assumed to be 
positive integers, $r$ dividing $n$. 
The notation $[n]$ denotes the set $\{ 1, 2, \dots, n\}$ and 
$\mathbb{S}_n$ denotes the set of permutations 
on $[n]$, i.e. the symmetric group of size $n!$.
For $\sigma \in \mathbb{S}_n$, 
we write $\sigma = [\sigma(1), \sigma(2), \dots, \sigma(n)]$, 
where for all $i \in [n],$ $\sigma(i)$ is the image of $i$ 
under $\sigma.$  
Throughout this paper we assume 
$\sigma$, $\pi \in \mathbb{S}_n$.  
With a slight abuse of notation, we may also use 
$\sigma$ to mean the sequence 
$(\sigma(1), \sigma(2), \dots, \sigma(n))\in \mathbb{Z}^n$
associated with $\sigma \in \mathbb{S}_n$.  
Multiplication of permutations is defined by composition so that 
for all $i \in [n]$, we have 
$(\sigma \tau)(i) = \sigma (\tau(i)).$ 
The identity permutation, $[1, 2, \dots, n] \in \mathbb{S}_n$ 
is denoted by $e$.

An $r$-\textit{regular multiset} is a multiset such that 
each of its elements is repeated $r$ times. 
A \textit{multipermutation} is an ordered tuple of the 
elements of a multiset, and in the instance of an $r$-regular multiset, is called 
an \textit{$r$-regular multipermutation}.  
Following the work of \cite{Farnoudtwo}, this study focuses 
on $r$-regular multipermutations, although many results 
are extendible to general multipermutations.

For each $\sigma \in \mathbb{S}_n$ we 
define a corresponding $r$-regular 
multipermutation $\mathbf{m}_\sigma^r$
as follows: 
for all $i \in [n]$ and $j \in [n/r],$ 
\[\mathbf{m}_\sigma^r(i):= j \text{ if and only if } (j-1)r + 1 \le \sigma(i) \le jr,\]
and $\mathbf{m}_\sigma^r :=
(\mathbf{m}_\sigma^r(1), \mathbf{m}_\sigma^r(2), \dots, \mathbf{m}_\sigma^r(n)) \in \mathbb{Z}^n$.
For example, if $n = 6,$ $r = 2$, and $\sigma = [1,5,2,4,3,6],$ then 
$\mathbf{m}_\sigma^r = (1,3,1,2,2,3).$ 
This definition differs slightly from the 
correspondence defined in \cite{Farnoudtwo}, which was 
defined in terms of the inverse permutation.  This is so that 
properties of the Ulam metric for permutations  
will carry over to the Ulam metric for multipermutations
(Lemmas \ref{n-l} and \ref{translocations} of 
Section \ref{multiperm}). 

With the correspondence above, we may define 
an equivalence relation between elements of $\mathbb{S}_n$. 
We say that $\sigma \equiv_r \pi$ if and only if 
$\mathbf{m}_\sigma^r = \mathbf{m}_\pi^r$. 
The equivalence class $R_r(\sigma)$ of $\sigma \in \mathbb{S}_n$ 
is defined as $R_r(\sigma) := \{\pi \in \mathbb{S}_n \;:\; \pi \equiv_r \sigma\}.$ 
For a subset $S \subseteq \mathbb{S}_n,$ the notation 
$\mathcal{M}_r(S) := \{\mathbf{m}_\sigma^r \; : \; \sigma \in S\},$
i.e. the set of $r$-regular multipermutations corresponding
to elements of $S$.

We say that $\sigma \equiv_r \pi$ if and only if 
$\mathbf{m}_\sigma^r = \mathbf{m}_\pi^r$. 
The equivalence class $R_r(\sigma)$ of $\sigma \in \mathbb{S}_n$ 
is defined by $R_r(\sigma) := \{\pi \in \mathbb{S}_n \;:\; \pi \equiv_r \sigma\}.$ 
For a subset $S \subseteq \mathbb{S}_n,$ the notation 
$\mathcal{M}_r(S) := \{\mathbf{m}_\sigma^r \; : \; \sigma \in S\},$
i.e. the set of $r$-regular multipermutations corresponding
to elements of $S$. 

The following definition is our own. 
For any $\mathbf{m} \in \mathbb{Z}^n$, and 
$\sigma \in \mathbb{S}_n,$ we define the product 
(a right group action) 
$\mathbf{m} \cdot \sigma$ by composition, 
similarly to the definition of multiplication 
of permutations. 
More precisely, for all $i \in [n]$, let 
$(\mathbf{m}\cdot \sigma)(i) := 
\mathbf{m}(\sigma(i))$. 
It is easily confirmed that 
$\mathbf{m} \cdot e = \mathbf{m}$
for all $\mathbf{m} \in \mathbb{Z}^n$. 
It is also easily confirmed that for all 
$\sigma, \pi \in \mathbb{S}_n$, we have
$\mathbf{m}\cdot(\sigma\pi) = (\mathbf{m}\cdot\sigma)\cdot\pi$. 
With this definition, notice that 
$\mathbf{m}_\sigma^r \cdot \pi = 
\mathbf{m}^r_{\sigma \pi}$.  
It is possible for different permutations to 
correspond to the same multipermutation, but 
for $\tau \in \mathbb{S}_n,$ it is clear that 
 $\mathbf{m}_\sigma^r = \mathbf{m}_\pi^r$ 
implies 
$\mathbf{m}_{\sigma}^r \cdot \tau = \mathbf{m}_{\pi}^r \cdot \tau.$

We finish this section by defining what 
a multipermutation code is.  
A subset $C \subseteq \mathbb{S}_n$ is called an 
$r$-regular multipermutation code if and only if 
for all $\sigma \in C,$ we also have $R_r(\sigma) \subseteq C.$ 
Such a code is denoted by $\mathsf{MPC}(n,r)$, and 
we say that $C$ is an $\mathsf{MPC}(n,r)$.  
If $C$ is an $\mathsf{MPC}(n,r)$ then 
whenever a permutation is a member of 
$C$  its entire equivalence 
class is also contained within $C$.  
Thus if $C$ is an  
$\mathsf{MPC}(n,r)$ it can be represented by 
the set of $r$-regular multipermutations associated 
with elements of $C$, i.e. the set 
$\mathcal{M}_r(C)$.   
Moreover, we define the cardinality $|C|_r$ 
of an $\mathsf{MPC}(n,r)$ $C$ as 
$|C|_r := |\mathcal{M}_r(C)|$ 
(this notation and definition differs slightly from 
\cite{Farnoudtwo}).

\section{Multipermutation Ulam Metric}\label{multiperm}
In this section we discuss some similarities and differences 
between the Ulam metric for permutations and the Ulam 
metric for multipermutations.  
We begin by defining the Ulam metric for permutations. 

For any two sequences $\mathbf{u}, \mathbf{v} \in \mathbb{Z}^n,$ 
$\ell(\mathbf{u},\mathbf{v})$ denotes the length of the longest 
common subsequence of $\mathbf{u}$ and $\mathbf{v}$. 
In other words,  
$\ell(\mathbf{u}, \mathbf{v})$ is the largest integer $k \in \mathbb{Z}_{>0}$ 
such that there exists a sequence $(a_1, a_2, \dots, a_k)$ 
where for all $p \in [k]$, we have $a_p = \mathbf{u}(i_p) = \mathbf{v}(j_k)$ 
with 
 $1\le i_1 < i_2 < \dots < i_k \le n$ and $1 \le j_1 < j_2 < \dots < j_k\le n$.
 The \textit{Ulam distance} $\mathrm{d}_\circ(\sigma, \pi)$ between permutations 
 $\sigma, \pi \in \mathbb{S}_n$ is defined as 
 $\mathrm{d}_\circ(\sigma, \pi) := n - \ell(\sigma, \pi).$ 
 
 It is also known that the Ulam distance $\mathrm{d}_\circ(\sigma, \pi)$ 
 between $\sigma, \pi \in \mathbb{S}_n$ is equivalent to the 
 minimum number of translocations needed to transform 
 $\sigma$ into $\pi$ \cite{Farnoud}.  Here, for distinct $i, j \in [n],$ 
 the translocation $\phi(i,j) \in \mathbb{S}_n$ 
 is defined as follows:

 \[
\phi (i,j) :=
\begin{cases}
  [1,2,\dots i-1,i+1,i+2, \dots, j, i, j+1,\dots, n]  \\
\hfill  \text{if } i < j  \\
  [1,2,\dots j-1, i, j, j+1, \dots, i-1, i+1, \dots n]   \\ 
\hfill  \text{otherwise } 
  \end{cases} 
\]

The notation $\phi(i,i)$ is understood to mean the identity 
permutation, $e.$ 
When it is not necessary to specify any index, a translocation 
may be written simply as $\phi$. 
Intuitively, when multiplied on the right of a permutation
 $\sigma \in \mathbb{S}_n,$ 
the translocation $\phi(i,j) \in \mathbb{S}_n$ deletes 
$\sigma(i)$ from the 
$i$th position of $\sigma$ and then inserts it in the new 
$j$th position
(shifting positions between $i$ and $j$ in the process).

The \textit{$r$-regular Ulam distance} $\mathrm{d}_\circ^r(\sigma, \pi)$ 
between permutations $\sigma, \pi \in \mathbb{S}_n$ 
is defined as the minimum Ulam distance among all 
members of $R_r(\sigma)$ and $R_r(\pi)$.  That is, 
$\mathrm{d}_\circ^r(\sigma, \pi) := 
\underset{\sigma' \in R_r(\sigma), \pi' \in R_r(\pi) }{\min}
\mathrm{d}_\circ(\sigma', \pi').$  
Notice that the $r$-regular Ulam distance is 
defined over equivalence classes.  

Although technically a distance between equivalence 
classes, it is convenient to think of the $r$-regular 
Ulam distance instead as a distance between 
multipermutations. 
Viewed this way, the property of the Ulam metric for 
permutations, that it can be defined in terms of 
longest common subsequences or equivalently 
in terms of translocations, carries over to the 
$r$-regular Ulam distance.  
The next lemma states that the $r$-regular Ulam 
distance between permutations $\sigma$ and $\pi$ 
is equal to $n$ minus the length of the longest 
common subsequence of their corresponding 
$r$-regular multipermutations.

\begin{lemma}\label{n-l}
$\mathrm{d}_\circ^r(\sigma, \pi) =
n - \ell(\mathbf{m}_\sigma^r, \mathbf{m}_\pi^r).$
\end{lemma}

\begin{proof}
We will first show that 
$\mathrm{d}_\circ^r(\sigma, \pi) \ge n - 
\ell(\mathbf{m}_\sigma^r, \mathbf{m}_\pi^r)$.
By definition of $\mathrm{d}_\circ^r(\sigma,\pi),$ there 
exist $\sigma' \in R_r(\sigma)$ and 
$\pi' \in R_r(\pi)$ such that $\mathrm{d}_\circ^r(\sigma, \pi) 
= \mathrm{d}_\circ(\sigma', \pi') = n - \ell(\sigma', \pi').$ 
Hence if 
for all $\sigma' \in R_r(\sigma)$ and $\pi' \in R_r(\pi)$ 
we have $\ell(\sigma', \pi') \le 
\ell(\mathbf{m}_\sigma^r, \mathbf{m}_\pi^r)$, 
then $\mathrm{d}_\circ^r(\sigma,\pi) 
\ge n - \ell(\mathbf{m}_\sigma^r, \mathbf{m}_\pi^r)$
(subtracting a larger value from $n$ 
results in a smaller overall value).
Therefore it suffices to show that 
that for all $\sigma' \in R_r(\sigma)$ and 
$\pi' \in R_r(\pi),$ that $\ell(\sigma', \pi') \le 
\ell(\mathbf{m}_\sigma^r, \mathbf{m}_\pi^r)$.  
This is simple to prove 
because if two permutations have a common subsequence, 
then their corresponding $r$-regular multipermutations 
will have a related common subsequence. 
Let $\sigma' \in R_r(\sigma)$, $\pi' \in R_r(\pi)$, and 
$\ell(\sigma', \pi') = k.$  Then there exist indexes 
$1 \le i_1 < i_2 < \dots < i_k \le n$ and 
$1 \le j_1 < j_2 < \dots < j_k \le n$ such that  
for all $p \in [k],$ $\sigma'(i_p) = \pi'(j_p).$  
Of course, whenever $\sigma'(i) = \pi'(j)$, then 
$\mathbf{m}_{\sigma'}^r(i) = \mathbf{m}_{\pi'}^r(j)$.
Therefore $\ell(\sigma', \pi') =k \le 
\ell(\mathbf{m}_{\sigma'}^r, \mathbf{m}_{\pi'}^r)
= \ell(\mathbf{m}_\sigma^r, \mathbf{m}_\pi^r).$

Next, we will show that $\mathrm{d}_\circ^r(\sigma, \pi) \le
n - \ell(\mathbf{m}_\sigma^r, \mathbf{m}_\pi^r).$
Note that
\begin{align*}
\mathrm{d}_\circ^r(\sigma, \pi) 
= \underset{\sigma' \in R_r(\sigma), \pi' \in R_r(\pi)}{\min}
\mathrm{d}_\circ(\sigma', \pi') \\ 
= \underset{\sigma' \in R_r(\sigma), \pi' \in R_r(\pi)}{\min} 
(n - \ell(\sigma', \pi')) \\
= n - \underset{\sigma' \in R_r(\sigma), \pi' \in R_r(\pi)}{\max}
\ell(\sigma', \pi').
\end{align*}
Here if $\underset{\sigma' \in R_r(\sigma), \pi' \in R_r(\pi)}{\max} 
\ell (\sigma', \pi') \ge
\ell(\mathbf{m}_\sigma^r, \mathbf{m}_\pi^r)$, 
then \\
$\mathrm{d}_\circ^r(\sigma,\pi) \le 
n - \ell(\mathbf{m}_\sigma^r, \mathbf{m}_\pi^r)$ 
(subtracting a smaller value from $n$ results in 
a larger overall value). 
It is enough to show that 
 there exist $\sigma' \in R_r(\sigma)$ and 
$\pi' \in R_r(\pi)$ such that 
$\ell(\sigma', \pi') \ge 
\ell(\mathbf{m}_\sigma^r, \mathbf{m}_\pi^r)$. 
To prove this fact, we take a longest common subsequence 
of $\mathbf{m}_\sigma^r$ and $\mathbf{m}_\pi^r$ and then 
carefully choose $\sigma' \in R_r(\sigma)$ and 
$\pi' \in R_r(\pi)$ to have an equally long common subsequence.  
The next paragraph describes how this can be done. 

Let $\ell(\mathbf{m}_\sigma^r, \mathbf{m}_\pi^r) = k$ and  
let $(1 \le i_1 < i_2 < \dots < i_k \le n)$ and 
$(1 \le j_1 < j_2 < \dots < j_k \le n)$ be integer sequences 
such that for all $p \in [k],$ 
$\mathbf{m}_\sigma^r(i_p) = \mathbf{m}_\pi^r(j_p).$ 
The existence of such sequences is guaranteed 
by the definition of $\ell(\mathbf{m}_\sigma^r, \mathbf{m}_\pi^r).$ 
Now for all $p \in [k],$ define 
$\sigma'(i_p)$ to be the smallest integer $l \in [n]$ 
such that $\mathbf{m}_\sigma(l) = \mathbf{m}_\sigma(i_p)$ 
and if $q \in [k]$ with $q < p,$ then 
$\mathbf{m}_\sigma^r(i_q) = \mathbf{m}_\pi^r(i_p)$ implies 
$\sigma'(i_q) < \sigma'(i_p) = l.$  
For all $p \in [k],$ define $\pi(j_p)$ similarly.  
Then for all $p \in [k],$ $\sigma'(i_p) = \pi'(j_p).$
The remaining terms of $\sigma'$ and $\pi'$ may 
easily be chosen in such a manner that 
$\sigma' \in R_r(\sigma)$ and $\pi' \in R_r(\pi).$
Thus there exist $\sigma' \in R_r(\sigma)$ and 
$\pi' \in R_r(\pi)$ such that 
$\ell(\sigma',\pi') \ge 
\ell(\mathbf{m}_\sigma^r, \mathbf{m}_\pi^r)$. 
\end{proof}

The following example helps to illuminate 
the choice of $\sigma'$ and $\pi'$ in the proof above.  
If $\mathbf{m}_\sigma^r = (2,1,2,1,3,3)$, and 
$\mathbf{m}_\pi^r = (3,2,2,1,3,1),$ then we 
have $\ell(\mathbf{m}_\sigma^r, \mathbf{m}_\pi^r) = 4,$
with the common subsequence $(2,2,1,3)$ of maximal length. 
Here $(1,3,4,6)$ and $(2,3,4,5)$ are 
sequences with 
$\mathbf{m}_\sigma^r(1) = \mathbf{m}_\pi^r(2)$,
$\mathbf{m}_\sigma^r(3) = \mathbf{m}_\pi^r(3)$, 
$\mathbf{m}_\sigma^r(4) = \mathbf{m}_\pi^r(4)$, and
$\mathbf{m}_\sigma^r(6) = \mathbf{m}_\pi^r(5).$ 
Then following the convention outlined in the proof above, 
$\sigma'(1) = \pi'(2) = 3$, 
$\sigma'(3) = \pi'(3) = 4$,
$ \sigma'(4) = \pi'(4) = 1$, and 
$\sigma'(6) = \pi'(5) = 5$, 
so that $\ell(\sigma',\pi') \ge 4$. 
The other elements of $\sigma'$ and $\pi'$ can 
be chosen as follows so that $\sigma' \in R_r(\sigma)$ and 
$\pi' \in R_r(\pi)$: 
set $\sigma'(2) = 1$, $\sigma'(5) = 6$, 
$\pi'(1) = 1$, and $\pi'(6) = 6.$

If two multipermutations $\mathbf{m}_\sigma^r$ 
and $\mathbf{m}_\pi^r$ have a common subsequence of 
length $k$, then $\mathbf{m}_\sigma^r$ can 
be transformed into $\mathbf{m}_\pi^r$ with $n-k$ (but 
no fewer) delete/insert operations.  
Delete/insert operations correspond to applying 
(multiplying on the right) a translocation.  
Hence by Lemma \ref{n-l} we can 
state the following lemma about the $r$-regular 
Ulam distance.  

\begin{lemma}\label{translocations}
$\mathrm{d}_\circ^r(\sigma, \pi) = 
\min \{ k \in \mathbb{Z}_{\ge0} \;:\; \text{there exists}  \; (\phi_1, \phi_2, \dots, \phi_k)\; 
s.t. \; \mathbf{m}_\sigma^r \cdot \phi_1  \phi_2 \cdots \phi_k = \mathbf{m}_\pi^r \}$.
\end{lemma}

\begin{proof}
There exists a translocation $\phi \in \mathbb{S}_n$ 
such that $\ell(\mathbf{m}_\sigma^r \cdot\phi, \mathbf{m}_\pi^r) 
= \ell(\mathbf{m}_\sigma^r, \mathbf{m}_\pi^r) + 1$, 
since it is always possible to arrange one element with a 
single translocation.  This then implies that 
$\min \{k \in \mathbb{Z} \;:\; \text{there exists } (\phi_1, \dots, \phi_k) \;s.t.\; 
\mathbf{m}_\sigma^r \cdot \phi_1 \cdots \phi_k = \mathbf{m}_\pi^r\} 
\le n - \ell(\mathbf{m}_\sigma^r, \mathbf{m}_\pi^r) = 
\mathrm{d}_\circ^r(\sigma, \pi).$
At the same time, given 
$\ell (\mathbf{m}_\sigma^r, \mathbf{m}_\pi^r) \le n,$ 
then for all translocations $\phi \in \mathbb{S}_n,$ 
we have that 
$\ell (\mathbf{m}_\sigma^r \cdot \phi, \mathbf{m}_\pi^r) 
\le \ell (\mathbf{m}_\sigma^r, \mathbf{m}_\pi^r) + 1$, 
since a single translocation can only arrange one 
element at a time.  Therefore 
$\min \{k \in \mathbb{Z} \;:\; \text{there exists } (\phi_1, \dots, \phi_k) \text{ s.t } 
\mathbf{m}_\sigma^r \cdot \phi_1 \cdots \phi_k = \mathbf{m}_\pi^r\} 
\ge n - \ell(\mathbf{m}_\sigma^r, \mathbf{m}_\pi^r) = 
\mathrm{d}_\circ^r(\sigma, \pi)$, by 
Lemma \ref{n-l}.
\end{proof}

Lemmas \ref{n-l} and \ref{translocations} allow us to view the 
Ulam metric for $r$-regular multipermutations 
similarly to the way we view the Ulam metric 
for permutations; in terms of longest common 
subsequences or in terms of the minimum number 
of translocations.  
Another known 
property of the Ulam metric for permutations is left 
invariance, i.e. given $
\tau \in \mathbb{S}_n,$ 
we have $\mathrm{d}_\circ(\sigma, \pi) = 
\mathrm{d}_\circ(\tau\sigma, \tau\pi).$  
However, left invariance does not hold in general 
for multipermutations, as the next lemma indicates.

\begin{lemma}
Let $n/r \ge 2$ and $r \ge 2.$
Then there exist $\sigma', \pi' \in \mathbb{S}_n$ such that 
$\mathrm{d}_\circ^r(e, \sigma') \ne \mathrm{d}_\circ^r(\pi' e,\pi'\sigma').$
\end{lemma}
\begin{proof}
Let 
$n/r \ge 2$ and $r \ge 2.$
Define $\sigma', \pi' \in \mathbb{S}_n$ by 
\begin{align*}
\sigma' := 
[\underset{2r}{\underbrace{2r, 2r-1, \dots, 1}}, 
\underset{n-2r}{\underbrace{2r+1, 2r+2, \dots, n}}] \text{ and }
\end{align*} 
\begin{align*} 
\pi' := 
[\underset{2r}{\underbrace{1, r+1, 2, r+2, \dots, r, 2r}}, 
\underset{n-2r}{\underbrace{2r+1, 2r+2, \dots, n}}].
\end{align*} 

First, consider $\mathrm{d}_\circ^r(e,\sigma').$ 
Note that for $\mathbf{m}_e^r$ and $\mathbf{m}_{\sigma'}^r,$ 
for any integer $i$ such that $2r < i \le n$ we have  
$e(i) = \sigma'(i)$, which implies $\mathbf{m}_e^r(i) = \mathbf{m}_{\sigma'}^r(i).$ 
Meanwhile, the first $2r$ elements of 
$\mathbf{m}_e^r$ and $\mathbf{m}_{\sigma'}^r$ are 
$(\underset{r}{\underbrace{1,1,\dots,1}},
\underset{r}{\underbrace{2,2,\dots,2}})$ and 
$(\underset{r}{\underbrace{2,2,\dots,2}},
\underset{r}{\underbrace{1,1,\dots,1}})$ respectively, 
so that the longest common subsequence of the first 
$2r$ elements of $\mathbf{m}_e^r$ and $\mathbf{m}_{\sigma'}^r$ 
is comprised of $r$ $1$'s or $r$ $2$'s.  
Hence $\ell(\mathbf{m}_e^r, \mathbf{m}_{\sigma'}^r) = 
(n-2r) + r = n-r$, which by lemma \ref{n-l} implies 
that $\mathrm{d}_\circ^r(e, {\sigma'}) = r \ge 2.$ 

Next, consider $\mathrm{d}_\circ^r(\pi' e, \pi' \sigma').$ 
Multiplying $\pi'$ and $\sigma'$ yields 
\begin{align*}
\pi'\sigma' = 
[\underset{2r}{\underbrace{2r, r, 2r-1, r-1, \dots, r+1, 1}}, 
\underset{n-2r}{\underbrace{2r+1, 2r+2, \dots, n}}].
\end{align*}
For all integers $i$ 
such that $2r < i \le n,$ we then have $\pi' e(i) = \pi'(i) = \pi'\sigma'(i) 
\implies  \mathbf{m}_{\pi' e}^r(i) = \mathbf{m}_{\pi'\sigma'}^r(i).$ 
Meanwhile, the first $2r$ elements of $\mathbf{m}_{\pi' e}^r$
and $\mathbf{m}_{\pi'\sigma'}$ are 
$(1,2,1,2,\dots,1,2)$ and $(2,1,2,1,\dots,2,1)$ respectively. 
Thus the longest common subsequence of the first $2r$ 
elements of $\mathbf{m}_{\pi' e}^r$ and $\mathbf{m}_{\pi'\sigma'}^r$ 
is any length $2r-1$ sequence of alternating $1$'s and $2$'s. 
Hence $\ell(\mathbf{m}_{\pi' e}^r, \mathbf{m}_{\pi'\sigma'}^r) 
= (n-2r) + (2r-1) = n-1,$ which by lemma \ref{n-l} implies 
that $\mathrm{d}_\circ^r(\pi' e, \pi' \sigma') = 1.$
\end{proof}

The fact that left invariance does not hold for 
the $r$-regular Ulam metric has implications 
on $r$-regular Ulam sphere sizes, defined and 
discussed in the next section.  
Left invariance implies sphere size does not depend upon
the center. However, we will demonstrate that in the 
multipermutation case sizes may differ 
depending upon the center, a 
fact previously unknown. 


\section{Young Tableaux Sphere Size Calculation}\label{young}
In \cite{self}, Young tableaux and the 
RSK-Correspondence were utilized to calculate 
Ulam Sphere sizes. A similar approach can be applied 
to $r$-regular Ulam spheres of arbitrary radius centered at 
$\mathbf{m}_e^r$.  
It is first necessary to introduce some basic notation and 
definitions regarding Young tableaux.  Additional information on the subject 
can be found in \cite{Fulton}, \cite{Schensted}, and \cite{Stanley}.

A \textit{Young diagram} is a left-justified collection of cells with a
(weakly) decreasing number of cells in each row below.  Listing the 
number of cells in each row gives a partition 
$\lambda = (\lambda_1, \lambda_2, \dots, \lambda_k)$ 
 of $n,$ where $n$ is the total
number of cells in the Young diagram.  
The notation $\lambda \vdash n$ indicates $\lambda$ is a partition of $n$.
Because the partition $\lambda \vdash n$ 
defines a unique Young diagram and vice versa, 
 a Young diagram may be referred to  by its associated 
 partition $\lambda \vdash n$.  
 For example, the partition $\lambda := (4,3,1) \vdash 8$ has the 
 corresponding Young diagram pictured below. 
 
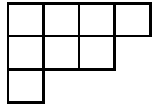
\[
\begin{Young}
   & & & \cr 
  & & \cr 
  \cr
\end{Young}
\]

A \textit{Young tableau}, or simply a \textit{tableau}, is 
a filling of a Young diagram $\lambda \vdash n$  
such that values in all cells are weakly increasing across 
each row and strictly increasing down each column. 
 If each of the integers 
$1$ through $n$ appears exactly once in a tableau $T$ that 
is a filling of a Young diagram $\lambda \vdash n,$ then we 
call $T$ a \textit{standard Young tableau}, abbreviated $SYT$.  

The \textit{Schensted algorithm} is an 
algorithm for obtaining a tableau $T\leftarrow x$ from 
a  tableau $T$ and a real number $x$.  
The algorithm may be defined as follows: \\~\\
$1)$ Set $i := 1$. \\
$2)$ If row $i$ (of $T$) is empty or if $x$ is greater or equal to 
     each of the entries in the $i$th row then input $x$ 
     in a new box at the end of row $i$
     and terminate the algorithm.
     Otherwise, proceed to step $3$. \\
$3)$ Find the minimum entry $y$ in row $i$ such that $y > x$ 
     and swap $y$ and $x$. That is, replace $y$ with $x$ in its box 
     and set $x := y.$   \\
$4)$ Set $i := i + 1$, and return to step $2$.  \\

As an example, 
let $T$ be the tableau pictured below on the far left. 
The following diagrams illustrate the stages of the
Schensted algorithm applied to $T$ to obtain $T \leftarrow 2.$ 
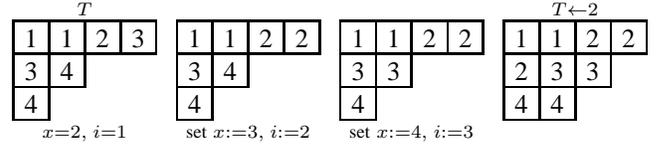
\[
\overset{T}{
\underset{x = 2,\; i = 1}{
\begin{Young}
1&1 &2 &3  \cr
 3&4 \cr 
 4 \cr
\end{Young}}}
\;\;\;
\underset{\text{set } x := 3, \; i := 2 }{
\begin{Young}
1&1&2 &2  \cr
 3 & 4 \cr 
 4 \cr
\end{Young}}
\;\;\;
\underset{\text{set } x := 4, \; i := 3}{
\begin{Young}
1 & 1 & 2 & 2 \cr
3 & 3  \cr
4 \cr
\end{Young}} 
\;\;\;
\overset{T \leftarrow 2}{
\begin{Young}
1 & 1 & 2 & 2 \cr
2 & 3 & 3 \cr
4 & 4\cr
\end{Young}}
\]

The Schensted algorithm may be applied to the sequence 
$(\mathbf{m}_\sigma^r(1),\mathbf{m}_\sigma^r(2), \dots, 
\mathbf{m}_\sigma^r(n))$ of an $r$-regular multipermutation $
\mathbf{m}_\sigma^r$ 
to obtain a unique tableau 
$P := (\dots(\mathbf{m}_\sigma^r(1) \leftarrow 
\mathbf{m}_\sigma^r(2))\leftarrow \dots )\leftarrow
\mathbf{m}_\sigma^r(n)$.
Meanwhile, a unique standard tableau results 
from recording where each new box appears in the construction 
of $P$.  This recording is accomplished by inputing the 
value $i$ in the new box that appears when $\mathbf{m}_\sigma^r(i)$ 
is added (via the Schendsted algorithm) to 
$(\dots(\mathbf{m}_\sigma^r(1)\leftarrow
(\mathbf{m}_\sigma^r(2) \leftarrow \dots) \leftarrow
\mathbf{m}_\sigma^r(i-1)$.  
Following conventions, we denote
 the standard tableau resulting from this recording 
 method by $Q$.
 As an example, the two tableaux  pictured below
 are the respective $P$ and $Q$ resulting from 
 the multipermutation $(2,3,2,1,3,1)$.  
 Notice that $P$ and $Q$ have the same shape.
 Intermediate steps are omitted for brevity.  
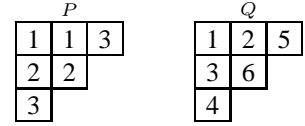
\[
\overset{P}{
\begin{Young}
1& 1 &3 \cr
  2&2 \cr 
 3\cr
\end{Young}}
\; \; \; \; \; \; \; \; \; \; 
\overset{Q}{
\begin{Young}
1&  2&  5\cr
3&6 \cr
4\cr
\end{Young}}
\]

The \textit{RSK-correspondence} (\cite{Fulton, Stanley}) 
provides a bijection between $r$-regular multipermutations
 $\mathbf{m}_\sigma^r$ and ordered pairs $(P,Q)$ on the same 
 Young diagram $\lambda \vdash n,$ where 
 $P$ is a tableaux whose members come from 
$\mathbf{m}_\sigma^r$ and $Q$ is a $SYT$.
A stronger form of the following lemma appears in \cite{Fulton}.

\begin{lemma}\label{length lemma}
Let $\sigma \in S_n$ and 
$P$ be the tableau resulting from running 
the Schendsted algorithm on the entries of $\sigma$.  
Then the number of columns in $P$ is equal to 
$\ell(\mathbf{m}_e^r,\mathbf{m}_\sigma^r)$, 
the length of the longest non-decreasing subsequence of 
$\mathbf{m}_\sigma^r$.
\end{lemma}

The above lemma, in conjunction with the RSK-correspondence, means that 
for all $k \in [n]$, the size of the set 
$\{\mathbf{m}_\sigma^r \in \mathcal{M}_r(\mathbb{S}_n) \;:\; 
\ell(\mathbf{m}_e^r,\mathbf{m}^r_\sigma) = k\}$ is equal to the 
sum of the number of ordered pairs 
$(P,Q)$ on each Young diagram $\lambda \vdash n$ such that 
$\lambda_1 = k,$ where $P$ is a tableaux whose members come from 
$\mathbf{m}_\sigma^r$ and $Q$ is a $SYT$. 
The number of $SYT$ on a particular $\lambda \vdash n$ is 
denoted by $f^\lambda.$ We denote by $K^\lambda_r$ 
(our own notation) the 
number of Young tableaux on $\lambda \vdash n$ such that 
each $i \in [n/r]$ appears exactly $r$ times. 
The next lemma states the relationship between 
$|S(\mathbf{m}_e^r, t)|$, $f^{\lambda}$, and $K^\lambda_r$.

\begin{lemma} \label{flambda applied}
Let 
$t \in [0,n-1]$, and 
$\Lambda := \{\lambda \vdash n \;:\; \lambda_1 \ge n-t\}.$
Then 
$|S (\mathbf{m}_e^r, t)| = \underset{\lambda \in \Lambda}{\sum} (f^{\lambda})(K^\lambda_r).$
\end{lemma}


\begin{proof}
Assume $t \in [0,n-1]$. 
Let $\Lambda := \{\lambda \vdash n \;:\; \lambda_1 \ge n-t\}.$
Furthermore, let 
$\Lambda^{(l)} := \{\lambda \vdash n \;:\; \lambda_1 = l\},$
the set of all partitions of $n$ 
having exactly $l$ columns.  
By the RSK-Correspondence, and Lemma \ref{length lemma}, 
there is a bijection between the set 
$\{\mathbf{m}_\sigma^r \;:\; \ell(\mathbf{m}_e^r, \mathbf{m}_\sigma^r) = l\}$ and 
the set of ordered pairs $(P,Q)$ where both $P$ and $Q$ have 
exactly $l$ columns.  
This implies that 
$|\{\mathbf{m}_\sigma^r \;:\; \ell(\mathbf{m}_e^r,\mathbf{m}_\sigma^r) = l\}| = 
\underset{\lambda \in \Lambda^{(l)}}{\sum} (f^{\lambda})(K^\lambda_r)$.
Note that by Lemma \ref{n-l}, 
\begin{eqnarray*}
|S (\mathbf{m}_e^r, t)| &=& 
|\{ \mathbf{m}_\sigma \;: \; \mathrm{d}_\circ^r(e,\sigma) \le t\}| \\
&=& |\{ \mathbf{m}_\sigma \;: \; \ell(\mathbf{m}_e^r,\mathbf{m}_\sigma^r)  
\ge n-t \}|.
\end{eqnarray*}
Hence it follows that 
$|S (\mathbf{m}_e^r, t)| 
= \underset{\lambda \in \Lambda}{\sum} (f^{\lambda})(K^\lambda_r).$
\end{proof}


The formula below, known as the 
\textit{hook length formula}, is due to Frame, Robinson, and Thrall \cite{Frame, Fulton}. 
In the formula,  the notation $(i,j) \in \lambda$
 is used to refer to the cell in the $i$th 
row and $j$th column of a Young diagram 
$\lambda \vdash n$.  The notation $h(i,j)$ denotes the
 \textit{hook length} of $(i,j) \in \lambda$, i.e., 
the number of boxes below or to the right of $(i,j)$, including the box $(i,j)$ 
itself.  More formally, 
$h(i,j) := \vert \{ (i,j^*) \in \lambda \;:\; j^* \ge j\} \cup \{ (i^*,j)
 \in \lambda \;:\; i^* \ge i\}\vert$.
 The formula is as follows:

\[
f^\lambda = \frac{n!}{\underset{(i,j)\in\lambda} \Pi h(i,j)}.
\]

Thus by applying Lemma \ref{flambda applied}, 
it is possible to calculate $r$-regular sphere size 
by using the hook length formula. 
We will use this strategy to treat the sphere of radius $r = 1$.  
However, because sphere sizes are calculated 
recursively, we must first calculate the sphere size 
when $r = 0.$

\begin{remark}\label{remark1}
$|S(\mathbf{m}_e^r, 0) = 1|$. 
Although this is an obvious fact, we wish to consider why it is true 
from the perspective of Lemma \ref{flambda applied}. 
Note first that there is only one partition $\lambda \vdash n$ 
such that $\lambda_1 = n$, namely 
$\lambda' := (n)$ with the associated Young diagram below.  
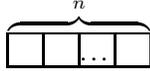
\[
\overbrace{
\begin{Young}
& & \dots &  \cr
\end{Young}
}^{n}
\]
It is clear that there is only one possible Young tableau on $\lambda'$ 
so that $(f^{\lambda'}) = 1,$ and thus by Lemma \ref{flambda applied} 
$\vert S(\mathbf{m}_e^r, 0) \vert$ = 1.
\end{remark}

The following proposition is an application of 
Lemma \ref{flambda applied}. 

\begin{proposition}\label{onesphere}
$|S (\mathbf{m}_e^r, 1)|  = 1 + (n-1)(n/r-1).$
\end{proposition}

\begin{proof}
First note that $|S (\mathbf{m}_e^r, 0)| = 1.$ 
  There is only one possible partition 
 $\lambda \vdash n$ such that $\lambda_1 = n-1,$
  namely $\lambda := (n-1,1),$ with its Young diagram pictured below. 
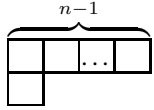
\[
\overbrace{
\begin{Young}
& & \dots &  \cr
 \cr
\end{Young}
}^{n-1}
\]

Therefore by Lemma \ref{flambda applied}, 
 $|S (\mathbf{m}_e^r, 1)| = 1 + (f^{\lambda'})(K^{\lambda'}_r).$  
Applying the well-known hook length formula 
(\cite{Frame, Fulton}), we obtain 
$f^{\lambda'} = n-1.$  The value $K^{\lambda'}_r$ is characterized 
by possible fillings of row $2$ with the stipulation that each $i \in [n/r]$
 must appear 
exactly $r$ times in the diagram.  In this case, since there is only a single 
box in row $2$, the possible fillings are $i \in [n/r-1],$ 
each of which yields a unique Young tableau 
of the desired type.  Hence $K^{\lambda'}_r = n/r -1$, 
which implies that $|S (\mathbf{m}_e^r, 1)| = 1 + (n-1)(n/r-1).$
\end{proof}

Proposition \ref{onesphere} demonstrates how 
Young Tableaux may be used to calculate $r$-regular 
Ulam spheres centered at $\mathbf{m}_e^r$.

\section{$r$-Regular Ulam Spheres and Duplication Sets}\label{upper}
In the previous section we showed how multipermutation 
Ulam spheres may be calculated when the center is 
$\mathbf{m}_e^r$. In this section we 
provide a way to calculate sphere sizes 
for any center when the radius is $t=1$. 
The $r$-regular Ulam sphere sizes play an 
important role in understanding the
potential code size for a given minimum distance. 

For example, the well-known sphere-packing 
and Gilbert-Varshamov bounds rely on 
calculating, or at least bounding sphere sizes.  
In the case of permutations, recall that the 
Ulam sphere $S(\sigma,t)$ centered at 
$\sigma$ of radius $t$ was defined as 
$S(\sigma,t) := \{\pi \in \mathbb{S}_n \;:\; 
\mathrm{d}_\circ(\sigma,\pi) \le t\}$, which is 
equivalent by definition to the set 
$\{\pi \in \mathbb{S}_n \;:\; 
n - \ell(\sigma, \pi) \le t\}$.

In the case of $r$-regular multipermutations, for 
$t \in \mathbb{Z}_{>0}$,
we introduce the following analogous definition 
of a sphere. 
\begin{definition}
Define 
\[
S(\mathbf{m}_\sigma^r, t) := \{
\mathbf{m}_\pi^r \in \mathcal{M}_r(\mathbb{S}_n) 
\;:\; \mathrm{d}_\circ^r(\sigma,\pi) \le t\}
\]
We call $S(\mathbf{m}_\sigma^r, t)$ 
the \textbf{$r$-regular Ulam sphere}  
centered at $\mathbf{m}_\sigma^r$ of radius $t$.
\end{definition} 

By Lemma \ref{n-l}, 
$S(\mathbf{m}_\sigma^r,t) = 
\{\mathbf{m}_\pi^r
 \in \mathcal{M}_r(\mathbb{S}_n) \;:\; 
 n - \ell(\mathbf{m}_\sigma^r, \mathbf{m}_\pi^r)  \le t\}$. 
 
It should be noted, however, that the notation $\mathbf{m}_\pi^r$
 is a bit misleading because given 
  $\mathbf{m}_\pi^r \in \mathcal{M}(\mathbb{S}_n),$ 
 we cannot uniquely determine $\pi.$  
The $r$-regular Ulam sphere definition can also be viewed 
in terms of translocations. Lemma \ref{translocations} implies 
that $S(\mathbf{m}_\sigma^r, t)$ is equivalent to 
 $\{\mathbf{m}_\pi^r 
 \in \mathcal{M}_r(\mathbb{S}_n) \;:\; 
\text{there exists } k \in [t] \text{ and } (\phi_1, \dots, \phi_k) \text{ s.t. } 
\mathbf{m}_\sigma^r \cdot \phi_1 \cdots \phi_k  = \mathbf{m}_\pi^r\}$.

Lemma \ref{flambda applied} provided a way to 
calculate $r$-regular Ulam spheres centered at 
$\mathbf{m}_e^r$. 
Unfortunately, the choice of center has an impact 
on the size of the sphere, as 
is easily confirmed by comparing Proposition \ref{onesphere} 
to Lemma \ref{omegasphere} (Section \ref{minmax}).
Hence the applicability of Lemma \ref{flambda applied}
 is limited.

We begin to address the issue of differing sphere 
sizes by considering the radius $t = 1$ case. 
To aid with calculating such sphere sizes, 
we introduce 
(as our own definition) 
the following subset of the set of translocations.  
\begin{definition} 
Let $n \in \mathbb{Z}_{>0}$. 
 Define 
 \[T_n := \{ \phi(i,j) \in \mathbb{S}_n \; : \; i - j \ne 1\}.\]
 We call $T_n$ the \textbf{unique set of translocations}.
\end{definition}
By definition, $T_n$ is the set of all translocations in 
$\mathbb{S}_n$, except 
translocations of the form $\phi(i,i-1)$.  We exclude translocations 
of this form because they can be modeled by translocations of the form 
$\phi(i-1,i)$, and are therefore redundant.  

We claim that the set $T_n$ is precisely the set of translocations 
needed to obtain all unique permutations within the 
Ulam sphere of radius $1$ via multiplication.  
Moreover, there is no redundancy in the set, 
that is, there is no smaller set of translocations 
yielding the entire Ulam sphere of radius $1$ 
when multiplied with a given center permutation. 
These facts are stated in the next lemma.

\begin{lemma}

$S(\sigma, 1) = \{\sigma \phi \in \mathbb{S}_n \;:\; \phi \in T_n\}$ 
and \\ $|T_n| = |S(\sigma, 1)|$. 

\end{lemma}

\begin{proof}
We will first show that 
$S(\sigma, 1) = \{\sigma \phi \in \mathbb{S}_n \;:\; \phi \in T_n\}$. 
Note that 
\begin{eqnarray*}
S(\sigma, 1) &=& 
\{\pi \in \mathbb{S}_n \;:\; \mathrm{d}_\circ(\sigma,\pi) \le 1\} \\
&=& \{\sigma \phi(i,j) \in \mathbb{S}_n \; : \; i,j \in [n]\}.
\end{eqnarray*}
It is trivial that 
\begin{eqnarray*}
T_n &=& \{\phi(i,j) \in \mathbb{S}_n \;:\; i - j \ne 1\} \\
&\subseteq& \{\phi(i,j) \in \mathbb{S}_n \;:\; i,j \in [n]\}.
\end{eqnarray*}
Therefore 
$\{\sigma\phi \in \mathbb{S}_n \;:\; \phi \in T_n\} \subseteq 
S(\sigma,1).$ 

To see why $S(\sigma, 1) \subseteq 
\{\sigma\phi \in \mathbb{S}_n \;:\; \phi \in T_n\}$, 
consider any $\sigma\phi(i,j) \in 
 \{\sigma \phi(i,j) \in \mathbb{S}_n \; : \; i,j \in [n]\} = 
S(\sigma,1).$  
If $i - j \ne 1,$ then $\phi(i,j) \in T_n,$ and thus 
$\sigma\phi(i,j) \in \{\sigma \phi \in \mathbb{S}_n \;:\; \phi \in T_n\}.$ 
Otherwise, if $i - j = 1,$ then $\sigma \phi(i,j) = \sigma \phi(j,i)$, and 
$i - j = 1 \implies j - i = -1 \ne 1,$ so $\phi(j,i) \in T_n.$  
Hence $\sigma\phi(i,j) = \sigma\phi(j,i) \in 
\{\sigma \phi \in \mathbb{S}_n \;:\; \phi \in T_n\}.$

Next we show that $|T_n| = |S(\sigma, 1)|$.
By Proposition \ref{onesphere}
(in the case that $r = 1$), $|S(\sigma,1)| = 1+ (n-1)^2$.  
On the other hand, $|T_n| = 
|\{\phi(i,j) \in \mathbb{S}_n \;:\; i - j \ne 1\}|.$  
If $i = 1,$ then there are $n$ values $j \in [n]$ 
such that $i - j \ne 1.$  Otherwise, if $i \in [n]$ but 
$i \ne 1,$ then there are $n-1$ values $j \in [n]$ such 
that $i - j \ne 1.$  However, for all $i, j \in [n],$ 
$\phi(i,i) = \phi(j,j) = e$ so that there are $n-1$ 
redundancies.  Therefore 
$|T_n| = n + (n-1)(n-1) - (n-1) = 1 + (n-1)^2.$    
\end{proof} 

In the case 
of permutations, the set 
$T_n$ has no redundancies.  If 
$\phi_1, \phi_2 \in T_n,$ then 
$\sigma\phi_1 = \sigma\phi_2$ implies $\phi_1 = \phi_2.$  
Alternatively, in the case of multipermutations, 
the set $T_n$ can generally be shrunken 
to exclude redundancies. 
Notice that 
$S(\mathbf{m}_\sigma^r, 1) 
= \{\mathbf{m}_\pi^r \in \mathcal{M}_r(\mathbb{S}_n) 
\;:\; 
\text{there exists } \phi \text{ s.t. } \mathbf{m}_\sigma^r \cdot \phi = \mathbf{m}_\pi\}$, which is equal to 
$\{\mathbf{m}_\sigma^r \cdot \phi \in \mathcal{M}_r(\mathbb{S}_n) 
\;:\; \phi \in T_n\}.$  
However, it is possible that there exist 
$\phi_1, \phi_2 \in T_n$ such that $\phi_1 \ne \phi_2,$ 
but $\mathbf{m}_\sigma^r \cdot \phi_1 = \mathbf{m}_\sigma^r \cdot \phi_2.$ 
In such an instance we may refer to either $\phi_1$ or $\phi_2$ 
as a \textbf{duplicate translocation} for $\mathbf{m}_\sigma^r$.  

If we remove all duplicate translocations for $\mathbf{m}_\sigma^r$ 
from $T_n$, then the resulting set will have the same 
cardinality as the $r$-regular Ulam sphere of radius $1$ centered 
at $\mathbf{m}_\sigma^r$. The next definition (our own) 
is the set of standard duplicate translocations. 
For the remainder of the paper, assume that 
$\mathbf{m}$ is an $n$-length integer tuple, i.e. 
$\mathbf{m} \in \mathbb{Z}^n$.

\begin{definition}
Define
\begin{eqnarray*}
D(\mathbf{m}) := 
\{ \; \; \; \phi(i,j) \in T_n\backslash\{e\} \;:\; 
(\mathbf{m}(i) = \mathbf{m}(j)) &&
\\
 \text{ or } (
\mathbf{m}(i) = \mathbf{m}(i-1) 
\text{ or } i =1)&
\}&
\end{eqnarray*}
We call $D(\mathbf{m})$ the \textbf{standard
duplicate translocation set} for $\mathbf{m}$. 
For each $i \in [n]$, also define 
$D_i(\mathbf{m}) := 
\{ \phi(i,j) \in D_n \; : \; j \in [n]\}$.
\end{definition}

If we take an $r$-regular multipermutation $\mathbf{m}_\sigma^r,$ 
then removing $D(\mathbf{m}_\sigma^r)$ from $T_n$ 
equates to removing a set 
of duplicate translocations.  
These duplications come in two varieties.  
The first variety corresponds to the 
first condition of the $D(\mathbf{m})$ definition, 
when $\mathbf{m}(i) = \mathbf{m}(j)$.
For example, if $\sigma \in \mathbb{S}_6$ such that 
$\mathbf{m}_\sigma^2 = (1,3,2,2,3,1)$, then we have 
$\mathbf{m}_\sigma^2\cdot \phi(1,5) = (3,2,2,3,1,1) = 
\mathbf{m}_\sigma^2 \cdot \phi(1,6),$ 
since $\mathbf{m}_\sigma^2(2) = 3 = \mathbf{m}_\sigma^2(4)$. 
This is because moving the first $1$ to the left or to the right 
of the last $1$ results in the same tuple. 

The second variety corresponds to the second condition of 
the $D(\mathbf{m})$ definition above, 
when $\mathbf{m}(i) = \mathbf{m}(i-1)$. 
For example, if 
$\mathbf{m}_\sigma^2 = (1,3,2,2,3,1)$ as before, then for all 
$j \in \{1,2,3,4,5,6\},$ we have 
$\mathbf{m}_\sigma^2 \cdot \phi(3,j) = \mathbf{m}_\sigma^2 \cdot \phi(4, j).$ 
This is because any translocation that deletes and inserts
 the second of the two adjacent $2$'s does not result in 
 a different tuple when compared to deleting and 
 inserting the first of the 
 two adjacent $2$'s.

\begin{lemma}\label{D} 
$S(\mathbf{m}_\sigma^r, 1) = 
\{\mathbf{m}_\sigma^r\cdot \phi \in \mathcal{M}_r(\mathbb{S}_n) \; :\; 
\phi \in T_n\backslash D(\mathbf{m}_\sigma^r)\}.$
\end{lemma}


\begin{proof}
Notice 
$S(\mathbf{m}_\sigma^r, 1) = 
\{\mathbf{m}_\sigma^r \cdot \phi \in \mathcal{M}_r(\mathbb{S}_n)
 \;:\; \phi \in T_n\}$. 
Hence it suffices to show that for all 
$\phi(i,j) \in D(\mathbf{m}_\sigma^r),$ there exists some 
$i', j' \in [n]$ such that $\phi(i',j') \in 
T_n\backslash D(\mathbf{m}_\sigma^r)$ and 
$\mathbf{m}_\sigma^r\cdot \phi(i,j) =
 \mathbf{m}_\sigma^r\cdot \phi(i',j').$ 
We proceed by dividing the proof into two 
main cases.  Case I is when 
$(\mathbf{m}_\sigma^r(i) \ne \mathbf{m}_\sigma^r(i-1)$ or $i=1$. 
Case II is when 
$(\mathbf{m}_\sigma^r(i) = \mathbf{m}_\sigma^r(i-1).$ 

Case I (when $(\mathbf{m}_\sigma^r(i) \ne \mathbf{m}_\sigma^r(i-1)$ or $i=1$)
can be split into two subcases: 
\begin{align*}
&& \text{Case IA: } &i < j &&\\
&& \text{Case IB: } &i > j. &&
\end{align*}

We can ignore the instance when $i = j$, since 
$\phi(i,j) \in D(\mathbf{m}_\sigma^r)$ implies 
$i \ne j.$  
For case IA, if for all $p \in [i,j]$
(for $a,b \in \mathbb{Z}$ with 
$a<b,$ the notation $[a,b] := \{a, a+1, \dots, b\}$) we have 
$\mathbf{m}_\sigma^r(i) = \mathbf{m}_\sigma^r(p)$, then 
$\mathbf{m}_\sigma^r \cdot \phi(i,j) = \mathbf{m}_\sigma^r \cdot e.$ 
Thus setting $i' = j' = 1$ yields the desired result. 
Otherwise, if there exists $p \in [i,j]$ such that 
$\mathbf{m}_\sigma^r(i) \ne \mathbf{m}_\sigma^r(p),$ 
then let 
\[
j^* := j - \min \{k \in \mathbb{Z}_{>0} \; : \; 
\mathbf{m}_\sigma^r(i) \ne \mathbf{m}_\sigma^r (j - k)\}.
\]
Then $\phi(i,j^*) \in T_n\backslash D(\mathbf{m}_\sigma^r)$ 
and $\mathbf{m}_\sigma^r \cdot \phi(i,j) 
= \mathbf{m}_\sigma^r \cdot \phi(i,j^*).$ 
Thus setting $i' = i$ and $j' = j^*$ yields the desired result. 
Case IB is similar to Case IA.  

Case II (when 
$\mathbf{m}_\sigma^r(i) = \mathbf{m}_\sigma^r(i-1)$), 
can also be divided into two subcases. 
\begin{align*}
&&\text{Case IIA: }& i < j &&\\
&&\text{Case IIB: }& i > j .&&
\end{align*}

As in Case I, we can ignore the instance 
when $i = j$. 
For Case IIA, if for all $p \in [i,j]$ we have 
$\mathbf{m}_\sigma^r(i) = \mathbf{m}_\sigma^r(p),$ 
then $\mathbf{m}_\sigma^r \cdot \phi(i,j) 
= \mathbf{m}_\sigma^r \cdot e,$
so setting $i = j = 1$ achieves the desired result. 
Otherwise, if there exists $p \in [i,j]$ such that 
$\mathbf{m}_\sigma^r(i) \ne \mathbf{m}_\sigma^r(p),$ 
then let 
\small
\[
i^* := i - \min 
\{k \in \mathbb{Z}_{>0} \;:\; 
(\mathbf{m}_\sigma^r(i) \ne \mathbf{m}_\sigma^r(i - k -1)) 
\text{ or } (i - k =1)\}.
\]  
\normalsize
Then 
$\mathbf{m}_\sigma^r \cdot \phi(i,j) 
= \mathbf{m}_\sigma^r \cdot \phi(i^*, j)$ 
and either one of the following is true: (1)
$\phi(i^*,j) \notin D_{i^*}(\mathbf{m}_\sigma^r) 
 \implies \phi(i^*,j) \notin D(\mathbf{m}_\sigma^r)$, so set 
 $i' = i^*$ and $j' = j$; or (2) 
 by Case IA there exist $i',j' \in [n]$ such that
  $\phi(i',j') \in T_n \backslash 
D(\mathbf{m}_\sigma^r)$ and
 $\mathbf{m}_\sigma^r \cdot \phi(i',j') 
 = \mathbf{m}_\sigma^r \cdot \phi(i^*,j) 
 = \mathbf{m}_\sigma^r \cdot \phi(i,j).$ 
 Case IIB is similar to Case IIA. 
\end{proof}

While Lemma \ref{D} shows 
that $D(\mathbf{m}_\sigma^r)$ is a set of duplicate translocations for 
$\mathbf{m}_\sigma^r,$ 
we have not shown 
that $T_n\backslash D(\mathbf{m}_\sigma^r)$ is the 
set of minimal size having the quality that 
$S(\mathbf{m}_\sigma^r,1) = 
\{\mathbf{m}_\sigma^r \cdot \phi \in \mathcal{M}_r(\mathbb{S}_n) \;:\; \phi \in 
T_n\backslash D(\mathbf{m}_\sigma^r)\}$.  
In fact it is not minimal. 
In some instances it is possible to remove 
further duplicate translocations to reduce the set size.  
We will define another set of duplicate translocations, but 
a few preliminary definitions are first necessary.

We say that $\mathbf{m}$ is 
\textit{alternating} if for all odd integers $1 \le i \le n$, 
 $\mathbf{m}(i) = \mathbf{m}(1)$ and 
for all even integers $2\le i'\le n$, 
$\mathbf{m}(i') = \mathbf{m}(2)$ but 
$\mathbf{m}(1) \ne \mathbf{m}(2)$. 
In other words, any alternating tuple is 
of the form 
$(a,b,a,b,\dots,a,b)$ or $(a,b,a,b,\dots,a)$ 
where $a,b \in \mathbb{Z}$ and 
$a \ne b$.
Any singleton is also said to be alternating. 
Now for integers 
$1 \le i\le n$ and  $0\le k \le n-i$, the \textit{substring} 
$\mathbf{m}[i, i + k]$ of $\mathbf{m}$ is defined as 
$\mathbf{m}[i, i + k ] := 
(\mathbf{m}(i), \mathbf{m}(i+1), \dots \mathbf{m}(i+k))$. 
Given a substring $\mathbf{m}[i,j]$ of $\mathbf{m}$, 
the \textit{length} $|\mathbf{m}[i,j]|$ of $\mathbf{m}[i,j]$ 
is defined as $|\mathbf{m}[i,j]| := j-i+1$. 
As an example, if $\mathbf{m}' : = (1,2,2,4,2,4,3,1,3)$, 
then $\mathbf{m}'[3,6] = (2,4,2,4)$ is an alternating 
substring of $\mathbf{m}'$ of length $4$.

\begin{definition}
Next define
 \begin{eqnarray*}
 E(\mathbf{m}) := 
\{ &&\phi(i,j) \in T_n \backslash D(\mathbf{m})  \; : \; \\
&& i < j \text { and there exists } \; k \in [i,j-2] \text{ s.t. } \\
&& (\phi(j,k) \in T_n\backslash D(\mathbf{m})) \\
 \text{ and } &&
(\mathbf{m}\cdot \phi(i,j) = \mathbf{m}\cdot\phi(j,k)) \; \; \;  \}.
\end{eqnarray*}

We call $E(\mathbf{m})$ the 
\textbf{alternating duplicate translocation set} for $\mathbf{m}$ 
because it is only nonempty when 
$\mathbf{m}$ contains an alternating substring 
of length at least $4$. 
For each $i \in [n],$  also define 
$E_i(\mathbf{m}) := 
 \{\phi(i,j) \in E(\mathbf{m}) \;:\; j \in [n]\}.$ 
\end{definition}
In the example of $\mathbf{m}' : = (1,2,2,4,2,4,3,1,3)$ above, 
$\mathbf{m}^* \cdot \phi(2,6) = 
\mathbf{m}' \cdot \phi(6,3)$ and 
$\phi(2,6),\phi(6,3) \in T_9 \backslash D(\mathbf{m}')$, 
implying that  $\phi(2,6) \in E(\mathbf{m}')$. 
In fact, it can easily be shown that 
$E(\mathbf{m}') = \{ \phi(2,6)\}$.

\begin{lemma}\label{remark2}
Let $i \in [n]$. 
Then 
$E_i(\mathbf{m}) \ne \varnothing$ 
if and only if \\ 
1) $\mathbf{m}(i) \ne \mathbf{m}(i-1)$ \\
2) There exists $j \in [i+1,n]$ and $k \in [i,j-2]$ such that 

i) For all $p \in [i,k-1]$, $\mathbf{m}(p) = \mathbf{m}(p+1)$ 

ii) $\mathbf{m}[k,j]$ is alternating 

iii) $|\mathbf{m}[k,j]| \ge 4$. 
\end{lemma}
\begin{proof}
Let $i \in [n]$. 
We will first assume 1) and 2) in the lemma statement 
and show that $E_i(\mathbf{m})$ is not empty. 
Suppose $\mathbf{m}(i) \ne \mathbf{m}(i-1)$, and 
that there exists $j \in [i+1,n]$ and $k \in [i,j-2]$ such 
that for all $p \in [i,k-1]$, we have 
$\mathbf{m}(p) = \mathbf{m}(p+1)$. Suppose also 
that $\mathbf{m}[k,j]$ is alternating with 
$|\mathbf{m}[k,j]| \ge 4.$ 

For ease of notation, let 
$a := \mathbf{m}(k) = \mathbf{m}(k+2)$ and 
$b := \mathbf{m}(k+1) = \mathbf{m}(k+3)$ 
so that $\mathbf{m}[k,k+3] = (a,b,a,b) \in \mathbb{Z}^4$.  
Then 
\begin{align*}
(\mathbf{m} \cdot \phi(i,k+3))[k,k+3] = & 
~ (\mathbf{m}\cdot \phi(k,k+3))[k,k+3] \\ 
= & ~ (b,a,b,a) \\
= & ~ (\mathbf{m}\cdot \phi(k+3,k))[k,k+3].
\end{align*} 
Moreover, for all $p \notin [k,k+3],$ we have 
$(\mathbf{m}\cdot\phi(i,k+3))(p) = \mathbf{m}(p) 
=  (\mathbf{m}\cdot\phi(k+3,k))(p)$. Therefore 
$\mathbf{m}\cdot\phi(i,k+3) = 
\mathbf{m}\cdot \phi(k+3,k)$.  Also notice 
that $\mathbf{m}(i) \ne \mathbf{m}(i-1)$ 
implies that $\mathbf{m}\phi(i,k+3) 
\notin D(\mathbf{m})$.  
Hence $\phi(i,k+3) \in E_i(\mathbf{m})$. 

We now prove the second half of the lemma. 
That is, we assume that $E_i(\mathbf{m}) 
\ne \varnothing$ and then show that 
1) and 2) necessarily hold. 
Suppose that $E_i(\mathbf{m})$ is nonempty. 
Then $\mathbf{m}(i) \ne \mathbf{m}(i-1)$, since 
otherwise there would not exist any 
$\phi(i,j) \in T_n \backslash D(\mathbf{m})$. 

Let $j \in [i+1,n]$ and $k \in [i,j-2]$ such that 
$\phi(j,k) \in T_n \backslash D(\mathbf{m}) 
\text{ and } \mathbf{m}\cdot\phi(i,j) = \mathbf{m}(j,k)$. 
Existence of such $j$, $k$, and $\phi(j,k)$ is 
guaranteed by definition of $E_i(\mathbf{m})$ and 
the fact that $E_i(\mathbf{m})$ 
was assumed to be nonempty. 
Then for all $p \in [i,k-1]$, we have 
$\mathbf{m}(p) = \mathbf{m}(p+1)$ and 
for all $p \in [k,j-2]$, we have 
$\mathbf{m}(p) = \mathbf{m}(p+2)$. 
Hence either $\mathbf{m}[k,j]$ is 
alternating, or else 
for all $p,q \in [k,j]$, we have 
$\mathbf{m}(p) = \mathbf{m}(q)$. 
However, the latter case is impossible, 
since it would imply that for all 
$p,q \in [i,j]$ that $\mathbf{m}(p) = \mathbf{m}(q)$, 
which would mean $\phi(j,k) \notin 
T_n \backslash D(\mathbf{m})$, a contradiction. 
Therefore $\mathbf{m}[k,j]$ is alternating. 

It remains only to show that $|\mathbf{m}[k,j]| \ge 4$. 
Since $k \in [i,j-2],$ it must be the case that 
$|\mathbf{m}[k,j]| \ge 3$. However, if 
$|\mathbf{m}[k,j]| = 3$ (which occurs when $k = j-2$), 
then $(\mathbf{m}\cdot\phi(i,j))(j) 
= \mathbf{m}(i) = \mathbf{m}(k) \ne \mathbf{m}(k+1) = 
(\mathbf{m}\cdot\phi(j,k)(j)$, 
which implies that 
$\mathbf{m}\cdot\phi(i,j) \ne \mathbf{m}\cdot\phi(j,k)$, 
a contradiction. 
Hence $|\mathbf{m}[k,j]| \ge 4$. 
\end{proof}

One implication of Lemma \ref{remark2} is that 
there are only two possible forms for 
$\mathbf{m}[i,j]$ where $\phi(i,j) \in E_i(\mathbf{m})$. 
The first possibility is that  
$\mathbf{m}[i,j]$ is an alternating substring 
of the form 
$(a,b,a,b,\dots,a,b)$ 
(here $a,b\in \mathbb{Z}$), 
so that $\mathbf{m}[i,j] \cdot \phi(i,j)$ 
is of the form $(b,a,b,a\dots,b,a)$. 
In this case, as long as $\mathbf{m}[i,j]| \ge 4$, 
then setting $k=i$ implies that 
$k \in [i,j-2]$, that 
$\phi(j,k) \in T_n\backslash D(\mathbf{m})$, and that 
$\mathbf{m}[i,j]\cdot\phi(i,j) = 
\mathbf{m}[i,j]\cdot\phi(j,k)$. 

The other possibility is that $\mathbf{m}[i,j]$ 
is of the form 
$(\underset{k}{\underbrace{a,a,a,\dots,a}},
\underset{n-k}{\underbrace{b,a,b,\dots,a,b}})$ 
(again $a,b\in \mathbb{Z}$), so that 
$\mathbf{m}[i,j]\cdot\phi(i,j)$ is of the form 
$(\underset{k-1}{\underbrace{a,\dots,a}},
\underset{n-k+1}{\underbrace{b,a,b,\dots,b,a}})$. 
Again in this case, as long as $|\mathbf{m}[i,j]| \ge 4,$ then 
$k \in [i,j-2]$, 
$\phi(j,k) \in T_n\backslash D(\mathbf{m})$, and 
$\mathbf{m}[i,j]\cdot\phi(i,j) = 
\mathbf{m}[i,j]\cdot\phi(j,k)$.

\begin{remark}\label{rem2}
If \\
1) $\mathbf{m}$ is alternating and \\
2) $n$ is even 

then $\mathbf{m}\cdot\phi(1,n) 
= \mathbf{m}\cdot\phi(n,1)$. 
\end{remark}

\begin{remark}\label{rem3}
If \\
1) $\mathbf{m}$ is alternating \\
2) $n \ge 3$ \\
3) $n$ is odd, 

then $\mathbf{m}\cdot\phi(1,n) \ne \mathbf{m}\cdot\phi(n,1)$. 
\end{remark}

To calculate $|E(\mathbf{m}_\sigma^r|$, we define a 
set of equal size that is easier to calculate. 

\begin{definition}
Define 
\begin{eqnarray*}
E^*(\mathbf{m}) := \{\; \; \; (i,j) \in [n]\times [n] &:&
(\mathbf{m}[i,j] \text{ is alternating}) \\ 
&&\text{ and } (|\mathbf{m}[i,j]| \ge 4) \\
&&\text{ and } (|\mathbf{m}[i,j]| \text{ is even})\; \; \; \}.
 \end{eqnarray*}
For each $i \in [n]$, also define 
$E_i^*(\mathbf{m}) := 
\{(i,j) \in E^*(\mathbf{m}) \;:\; j \in [n]\}.$ 
Notice that $E^*(\mathbf{m})
= \underset{i \in [n]}{\bigcup} E_i^*(\mathbf{m})$. 
\end{definition}

\begin{lemma}\label{E*}
$|E(\mathbf{m})| = |E^*(\mathbf{m})|$
\end{lemma}
\begin{proof}
The idea of the proof is simple. 
Each element $\phi(i,j)\in E(\mathbf{m})$ 
involves exactly one alternating sequence 
of length greater or equal to $4$, so the set
sizes must be equal. We formalize the 
argument by showing that 
$|E(\mathbf{m})| \le |E^*(\mathbf{m})|$ 
and then that 
$|E^*(\mathbf{m})| \le |E(\mathbf{m})|$.

To see why 
$|E(\mathbf{m})| \le |E^*(\mathbf{m})|$, 
we define a mapping 
$f : [n] \to [n]$, which maps 
index values either to the 
beginning of the nearest alternating 
subsequence to the right, or else 
to $n$. For all $i \in [n]$, let

\footnotesize
 \begin{align*}
f(i) :=
\begin{cases}
  i + \min\{p \in \mathbb{Z}_{\ge 0} \;:\; 
    (\mathbf{m}(i) \ne \mathbf{m}(i+p+1))
     \vee   (i+p = n) \}  \\
\hfill (\text{if } \mathbf{m}(i) \ne \mathbf{m}(i-1) \text{ or } i = 1) \\
 n
\hfill  \text{(otherwise)} 
  \end{cases} 
\end{align*}
\normalsize

Notice by definition of $f$, if 
$i,i' \in [n]$ such that $i \ne i'$, and 
if $\mathbf{m}(i) \ne \mathbf{m}(i-1)$ or $i=1$ 
and at the same time 
$\mathbf{m}(i') \ne \mathbf{m}(i'-1)$ or $i'=1$, 
then $f(i) \ne f(i')$. 

Now for each $i \in [n]$, 
if $\mathbf{m}(i) \ne \mathbf{m}(i-1)$ or 
$i=1$, then 
$|E_i(\mathbf{m})|$ = $|E_{f(i)}^*(\mathbf{m})|$ 
by Lemma \ref{remark2} and the two previous remarks.
Otherwise, if
$\mathbf{m}(i) = \mathbf{m}(i-1)$, 
then 
$|E_i(\mathbf{m})| = |E_{f(i)}^*(\mathbf{m})| = 0$. 
Therefore 
$|E_i(\mathbf{m})| \le |E_i^*(\mathbf{m})|$. 
This is true for all $i \in [n]$, so 
$|E(\mathbf{m})| \le |E^*(\mathbf{m})|$.

The argument to show that 
$|E^*(\mathbf{m})| \le |E(\mathbf{m})|$ 
is similar, except it uses the following 
function $g : [n] \to [n]$ instead of $f$.  
For all $i \in [n]$, let 

\footnotesize
 \begin{align*}
g(i) :=
\begin{cases}
  i - \min\{p \in \mathbb{Z}_{\ge 0} \;:\; 
    (\mathbf{m}(i) \ne \mathbf{m}(i-p-1))
     \vee   (i-p = 1) \}  \\
\hfill (\text{if } \mathbf{m}(i) \ne \mathbf{m}(i-1) \text{ or } i = n) \\
 n
\hfill  \text{(otherwise)} 
  \end{cases} 
\end{align*}

\normalsize

\end{proof}

By definition, calculating $|E^*(\mathbf{m})|$ equates 
to calculating the number of alternating substrings 
$\mathbf{m}[i,j]$ of $\mathbf{m}$ such that the length 
of the substring is both even and longer than 4. The 
following lemma helps to simplify this calculation further. 

\begin{lemma}\label{oddeven}
Let $\mathbf{m}$ be an alternating string. Then 
\begin{align*}
&\text{1) If $n$ is even then } |E^*(\mathbf{m})| = 
\left(\frac{n-2}{2} \right)^2\\ 
&\text{2) If $n$ is odd then } |E^*(\mathbf{m})| = 
\left(\frac{n-3}{2}\right)\left(\frac{n-1}{2}\right)
\end{align*}
\end{lemma}

\begin{proof}
Assume $\mathbf{m}$ is an alternating string. 
By Lemma \ref{E*}, 
$|E(\mathbf{m})| = |E^*(\mathbf{m})| =
|\underset{i\in[n]}{\bigcup} E_i^*(\mathbf{m})|$. 
Since $\mathbf{m}$ was assumed to be alternating, 
\begin{eqnarray*}
& & |\underset{i \in [n]}{\bigcup} E_i^*(\mathbf{m})| \\
&=& |\{(i,j) \in [n] \times [n] \;:\; |\mathbf{m}[i,j]| \ge 4 
\text{ and } |\mathbf{m}[i,j]| \text{ is even} \}| \\
&=& |\{(i,j) \in [n]\times[n] \;:\; j-i+1 \in K\}|,
\end{eqnarray*}
where $K$ is the set of even integers between 
$4$ and $n$, i.e. 
$K := \{k \in [4,n] \;:\; k \text{ is even}\}$. 
For each $k \in K$, we have 
\begin{eqnarray*} 
& &|\{(i,j) \in [n] \times [n] \;:\; j-i+1 = k\}| \\
&=&|\{i \in [n] \;:\; i \in [1,n-k+1]\}| \\
&=& n - k + 1. 
\end{eqnarray*}

Therefore 
$|E(\mathbf{m})| = \underset{k \in K}{\sum}(n-k+1)$. 
In the case that $n$ is even, then 
\begin{eqnarray*}
\underset{k \in K}{\sum}(n-k+1) 
&=& \underset{i=2}{\overset{n/2}{\sum}}(n-2i+1) \\ 
&=& \underset{i=1}{\overset{(n-2)/2}{\sum}}(2i-1) 
= \left(\frac{n-2}{2} \right)^2.
\end{eqnarray*}
In the case that $n$ is odd, then 
\begin{eqnarray*}
\underset{k \in K}{\sum}(n-k+1) 
&=& \underset{i=2}{\overset{(n-1)/2}{\sum}}(n-2i+1) \\ 
&=& \underset{i=1}{\overset{(n-3)/2}{\sum}}2i 
= \left(\frac{n-3}{2} \right)
\left(\frac{n-1}{2}\right).
\end{eqnarray*}

\end{proof}

Notice that by Lemma \ref{oddeven}, 
it suffices to calculate 
$|E(\mathbf{m})|$ for locally maximal length 
alternating substrings of $\mathbf{m}$. 
An alternating substring $\mathbf{m}[i,j]$ 
is of \textit{locally maximal length} if and only if 
1) $\mathbf{m}[i-1]$ is not alternating or $i=1$; and 
2) $\mathbf{m}[i,j+1]$ is not alternating or $j = n$.

Finally, we define the general set of duplications. The 
lemma that follows the definition also shows that 
removing the set $D^*(\mathbf{m}_\sigma^r)$ from $T_n$ 
removes all duplicate translocations 
associated with $\mathbf{m}_\sigma^r$. 
\begin{definition}[$D^*(\mathbf{m})$, duplication set]
Define 
\[D^*(\mathbf{m}) := D(\mathbf{m}) \cup E(\mathbf{m}).\]
We call $D^*(\mathbf{m})$ the \textbf{duplication 
set} for $\mathbf{m}$. For each $i \in [n]$, we also define 
$D_i^*(\mathbf{m}) := 
\{ \phi(i,j) \in D^*(\mathbf{m}) \; : \; j \in [n]\}$.
\end{definition}

\begin{lemma}\label{D*}
Let $\phi_1, \phi_2 \in T_n \backslash 
D^*(\mathbf{m}_\sigma^r) $. 
Then $\phi_1 = \phi_2$ if and only if  
$\mathbf{m}_\sigma^r\cdot \phi_1 = \mathbf{m}_\sigma^r \cdot \phi_2.$ 
\end{lemma}

\begin{proof}
Let 
$\phi_1, \phi_2 \in T_n \backslash D^*(\mathbf{m}_\sigma^r)$. 
If $\phi_1 = \phi_2$ then
$\mathbf{m}_\sigma^r \cdot \phi_1 = \mathbf{m}_\sigma^r \cdot \phi_2$ 
trivially. It remains to prove that 
$\mathbf{m}_\sigma^r \cdot \phi_1 = \mathbf{m}_\sigma^r \cdot \phi_2 
\implies \phi_1 = \phi_2.$ We proceed by contrapositive. 
Suppose that $\phi_1 \ne \phi_2.$  
We want to show that 
$\mathbf{m}_\sigma^r\cdot \phi_1 \ne \mathbf{m}_\sigma^r\phi_2.$ 
Let $\phi_1 := \phi(i_1,j_1)$ and $\phi_2 := \phi(i_2, j_2)$.  
The remainder of the proof can be split into two main cases: 
Case I is if $i_1 = i_2$ and Case II is if $i_1 \ne i_2.$

Case I (when $i_1 = i_2$), can be further divided into two subcases: 
\small
\begin{flalign*}
&& \text{Case IA: }&
\mathbf{m}_\sigma^r(i_1) = 
\mathbf{m}_\sigma^r(i_1 -1)&&\\
&&  \text{Case IB: } &
 \mathbf{m}_\sigma^r(i_1) \ne
\mathbf{m}_\sigma^r(i_1 -1). &&
\end{flalign*}
\normalsize

Case IA is easy to prove.  We have  
$D_{i_1}^*(\mathbf{m}_\sigma^r) = 
D_{i_2}^*(\mathbf{m}_\sigma^r) = 
\{\phi(i_1,j) \in T_n\backslash \{e\} \;:\; j \in [n]\}$, 
so $\phi_1 = e = \phi_2,$ a contradiction.  
For Case IB, we can first assume without loss of 
generality that $j_1 < j_2$ and then split 
into the following smaller subcases:  
\small
\begin{align*}
&&\text{i) }& (j_1 < i_1) \text{ and } (j_2 > i_1)&& \\
&&\text{ii) }& (j_1 < i_1) \text{ and } (j_2 \le i_1)&& \\
&&\text{iii) }& (j_1 > i_1) \text{ and } (j_2 > i_1)&& \\ 
&&\text{iv) }& (j_1 > i_1) \text{ and } (j_2 \le i_1).&&
\end{align*}
\normalsize

However, subcase iv) is unnecessary since 
it was assumed that $j_1 < j_2,$ 
so $j_1 > i_1 \implies j_2 > j_1 > i_1.$  
Subcase ii) can also be reduced to 
$(j_1 < i_1) \text{ and } (j_2 < i_1)$ since 
$j_2 \ne i_2 = i_1.$  
Each of the remaining subcases is proven 
by noting that there is some element in the 
multipermutation $\mathbf{m}_\sigma^r \cdot \phi_1$ that is 
necessarily different from $\mathbf{m}_\sigma^r\cdot \phi_2.$ 
For example, in subcase i), we have 
$\mathbf{m}_\sigma^r\cdot \phi_1(j_1) = 
\mathbf{m}_\sigma^r(i_1) \ne 
\mathbf{m}_\sigma^r(j_1) = \mathbf{m}_\sigma^r\cdot\phi_2(j_1).$ 
Subcases ii) and iii) are solved similarly. \\

Case II (when $i_1 \ne i_2$) can be divided into three subcases: 
\small
\begin{flalign*}
&&\text{Case}& \text{ IIA: }
(\mathbf{m}_\sigma^r(i_1) = \mathbf{m}_\sigma^r(i_1 -1) 
\text{ and } 
\mathbf{m}_\sigma^r(i_2) = \mathbf{m}_\sigma^r(i_2-1)), &&\\ 
&&\text{Case}& \text{ IIB: either } && \\
&& &(\mathbf{m}_\sigma^r(i_1) = \mathbf{m}_\sigma^r(i_1 -1) 
\text{ and } 
\mathbf{m}_\sigma^r(i_2) \ne \mathbf{m}_\sigma^r(i_2-1)) &&\\
&& & \text{ \hspace{-6.8mm} or } 
(\mathbf{m}_\sigma^r(i_1) \ne \mathbf{m}_\sigma^r(i_1 -1) 
\text{ and } 
\mathbf{m}_\sigma^r(i_2) = \mathbf{m}_\sigma^r(i_2-1)), &&\\ 
&&\text{Case}& \text{ IIC: }
(\mathbf{m}_\sigma^r(i_1) \ne \mathbf{m}_\sigma^r(i_1 -1) 
\text{ and } 
\mathbf{m}_\sigma^r(i_2) \ne \mathbf{m}_\sigma^r(i_2-1)).  &&
\end{flalign*}
\normalsize

Case IIA is easily solved by mimicking the proof of Case IA.  
Case IIB is also easily solved as follows.  
First, without loss of generality, 
we assume that 
$\mathbf{m}_\sigma^r(i_1) = \mathbf{m}_\sigma^r(i_1-1) \text{ and } 
\mathbf{m}_\sigma^r(i_2) \ne \mathbf{m}_\sigma^r(i_2-1).$  
Then $D_{i_1}^*(\mathbf{m}_\sigma^r) = 
\{\phi(i_1,j)\in T_n \backslash \{e\} \; : \; 
j \in [n]\},$ so $\phi_1 = e.$  
Therefore we have 
$\mathbf{m}_\sigma^r\cdot \phi_1(j_2) = \mathbf{m}_\sigma^r(j_2) 
\ne \mathbf{m}_\sigma^r(i_2) = \mathbf{m}_\sigma^r\phi_2(i_2-1).$ \\

Finally, for Case IIC, without loss of generality we may 
assume that $i_1 < i_2$ and then split into the following four subcases: 
\small
\begin{align*}
&&\text{i) } &(j_1 < i_2) \text{ and } (j_2 \ge i_2) &&\\
&&\text{ii) } &(j_1 < i_2) \text{ and } (j_2 < i_2) &&\\
&&\text{iii) } &(j_1 \ge i_2) \text{ and } (j_2 \ge i_2) &&\\
&&\text{iv) } &(j_1 \ge i_2) \text{ and } (j_2 < i_2). && 
\end{align*}. 
\normalsize

However, since 
$\phi(i_2,j_2) \in T_n\backslash D^*(\mathbf{m}_\sigma^r)$ 
implies $i_2 \ne j_2,$ 
subcases i) and iii) can be reduced to 
$(j_1 < i_2) \text{ and } (j_2 > i_2)$ and 
$(j_1 \ge i_2) \text{ and } (j_2 > i_2)$ respectively.  
For subcase i), we have 
$\mathbf{m}_\sigma^r\cdot \phi_1(j_1) = \mathbf{m}_\sigma^r(i_1) 
\ne \mathbf{m}_\sigma^r(j_1) = \mathbf{m}_\sigma^r \cdot \phi_2(j_1).$ 
Subcases ii) and iii) are solved in a similar manner.  
For subcase iv), 
if $j_1 > i_2$, then $\mathbf{m}_\sigma^r \cdot \phi_1(j_1) = 
\mathbf{m}_\sigma^r (i_1) \ne \mathbf{m}_\sigma^r(j_1) = 
\mathbf{m}_\sigma^r \cdot \phi_2(j_1).$  
Otherwise, if $j_1 = i_2$, then 
$\phi_1 = \phi(i_1, i_2)$ and 
$\phi_1 = \phi(i_2, j_2).$ 
Thus if $\mathbf{m}_\sigma^r\cdot \phi_1 = 
\mathbf{m}_\sigma^r \cdot \phi_2$ then 
$\phi_1 \in D^*_{i_1}(\mathbf{m}_\sigma^r),$ 
which implies that 
$\phi_1 \notin T_n \backslash D^*(\mathbf{m}_\sigma^r)$, 
a contradiction. 
\end{proof}

Lemma \ref{D*} implies that we can 
calculate $r$-regular Ulam sphere sizes 
of radius $1$ whenever we can calculate the 
appropriate duplication set. This calculation 
can be simplified by noting that for a sequence 
$\mathbf{m} \in \mathbb{Z}^n$ that 
$D(\mathbf{m}) \cap E(\mathbf{m}) = \varnothing$ 
(by the definition of $E(\mathbf{m}))$ 
and then decomposing the duplication set into these 
components. 
This idea is stated in the next theorem

\begin{theorem}\label{ballcalc}
$|S(\mathbf{m}_\sigma^r,1)| 
= 1 + (n-1)^2  - |D(\mathbf{m}_\sigma^r)| - |E(\mathbf{m}_\sigma^r)|.$
\end{theorem}

\begin{proof}
By the definition of $D^*(\mathbf{m}_\sigma^r)$ and 
lemma \ref{D}, 
\begin{eqnarray*}
&&\{\mathbf{m}_\sigma^r \cdot \phi \in \mathcal{M}_r(\mathbb{S}_n) :
\phi \in T_n\backslash D^*(\mathbf{m}_\sigma^r)\} \\
&=&
\{\mathbf{m}_\sigma^r\cdot \phi \in \mathcal{M}_r(\mathbb{S}_n) \; :\; 
\phi \in T_n\backslash D(\mathbf{m}_\sigma^r)\} \\
&=& S(\mathbf{m}_\sigma^r,1).
\end{eqnarray*}
This implies $|T_n\backslash D^*(\mathbf{m}_\sigma^r)| 
\ge |S(\mathbf{m}_\sigma^r,1)|.$
By lemma \ref{D*}, for $\phi_1,\phi_2 \in 
T_n\backslash D^*(\mathbf{m}_\sigma^r),$ 
if $\phi_1 \ne \phi_2,$ then $\mathbf{m}_\sigma^r\cdot \phi_1 
\ne \mathbf{m}_\sigma^r\cdot\phi_2.$ 
Hence we have $|T_n\backslash D^*(\mathbf{m}_\sigma^r)|
\le |S(\mathbf{m}_\sigma^r,1)|,$ which implies that 
$|T_n\backslash D^*(\mathbf{m}_\sigma^r)|
= |S(\mathbf{m}_\sigma^r,1)|.$
It remains to show that 
$|T_n\backslash D^*(\mathbf{m}_\sigma^r)| 
= (n-1)^2 + 1 - |D(\mathbf{m}_\sigma^r)| - |E(\mathbf{m}_\sigma^r)|.$ 
This is an immediate consequence of the fact that 
$|T_n| = (n-1)^2 + 1$ and 
$D(\mathbf{m}_\sigma^r) \cap E(\mathbf{m}_\sigma^r) = \varnothing$.
\end{proof}

Theorem \ref{ballcalc} reduces the calculation of 
$|S(\mathbf{m}_\sigma^r,1)|$ to calculating 
$|D(\mathbf{m}_\sigma^r)|$ and $|E(\mathbf{m}_\sigma^r|$. 
It is an easy matter to calculate $|D(\mathbf{m}_\sigma^r)|$, 
since it is exactly equal to $(n-2)$ times the number of 
$i \in [n]$ such that 
$\mathbf{m}_\sigma^r(i) = \mathbf{m}_\sigma^r(i-1)$ plus 
$(r-1)$ times the number of $i \in [n]$ such that 
$\mathbf{m}_\sigma^r(i) \ne \mathbf{m}_\sigma^r(i-1)$. 
We also showed how to calculate $|E(\mathbf{m})|$ earlier. 
The next example is an application of Theorem \ref{ballcalc} 

\begin{example}
Suppose $\sigma := [1,2,3,4,9,6,7,11,5,10,12,8].$  Then 
$\mathbf{m}_\sigma^3 := (1,1,1,2,3,2,3,2,4,4,3,4)$. 
There are 3 values of $i \in [12]$ 
such that $\mathbf{m}_\sigma^3(i) = \mathbf{m}_\sigma^3(i-1)$, 
which implies that $|D(\mathbf{m}_\sigma^3| = 
(3)(12-2) + (12-3)(3-1) = 48$. Meanwhile, by 
Lemmas \ref{E*} and \ref{oddeven}, 
$|E(\mathbf{m}_\sigma^3)| = ((5-3)/2)((5-1)/2)) = 2$. 
By Theorem \ref{ballcalc}, 
 $|S(\mathbf{m}_\sigma^3), 1| = (12-1)^2 - 48 - 2 = 71$. 
\end{example}

\section{Min/Max Spheres and Code Size Bounds}\label{minmax} 
In this section we show choices of center achieving 
 minimum and maximum $r$-regular Ulam sphere sizes
for the radius $t=1$ case. The minimum and maximum values 
are explicitly given. We then discuss 
resulting bounds on code size. First let us consider 
the $r$-regular Ulam sphere of minimal size.

\begin{lemma}\label{esphere}
$|S(\mathbf{m}_e^r, 1)| \le |S(\mathbf{m}_\sigma^r,1)|$
\end{lemma}

\begin{proof}
In the case that $n/r = 1,$ then 
$\mathbf{m}_e^r = e$ and $\mathbf{m}_\sigma^r = \sigma$, 
so that 
$|S(\mathbf{m}_e^r, 1)| = |S(\mathbf{m}_\sigma^r,1)|$. 
Therefore we may assume that $n/r \ge 2$. 
By Theorem \ref{ballcalc}, 
$\underset{\sigma \in \mathbb{S}_n}{\min}
(|S(\mathbf{m}_\sigma^r,1)|)  
= 1 + (n-1)^2 - 
\underset{\sigma \in \mathbb{S}_n}{\max} 
(|D(\mathbf{m}_\sigma^r)| + |E(\mathbf{m}_\sigma^r)|). 
$
Since $n/r \ge 2,$ we know that 
$n-2 > r-1,$ which implies that for all 
$\sigma \in \mathbb{S}_n$, that 
$|D(\mathbf{m}_\sigma^r)|$ is maximized by 
maximizing the number of integers $i \in [n]$ such that 
$\mathbf{m}_\sigma^r(i) = \mathbf{m}_\sigma^r(i-1)$. 
This is accomplished by choosing $\sigma = e,$ and 
hence for all 
$\sigma \in \mathbb{S}_n$, we have 
$|D(\mathbf{m}_e^r)| \ge |D(\mathbf{m}_\sigma^r)|$.  

We next will show that for any increase in the size of 
$|E(\mathbf{m}_\sigma^r)|$ compared to 
$|E(\mathbf{m}_e^r)|$, that $|D(\mathbf{m}_\sigma^r)|$ 
is decreased by a larger value compared to 
$|D(\mathbf{m}_e^r)|$, so that 
$(|D(\mathbf{m}_\sigma^r)| + |E(\mathbf{m}_\sigma^r)|)$ is 
maximized when $\sigma = e$. 

Suppose $\sigma \in \mathbb{S}_n$. 
By Lemmas \ref{E*} and \ref{oddeven}, 
$|E(\mathbf{m}_\sigma^r|$ is characterized by 
the lengths of its locally maximal alternating substrings.  
For every locally maximal alternating substring 
$\mathbf{m}_\sigma^r[a,a+k-1]$ of 
$\mathbf{m}_\sigma^r$ of length $k$, 
there are at least $k-2$ fewer instances where 
$\mathbf{m}_\sigma^r = \mathbf{m}_\sigma^r(i-1)$ 
when compared to instances where 
$\mathbf{m}_e^r(i) = \mathbf{m}_e^r(i-1)$. 
This is because for all $i \in [a+1,a+k-1]$, 
$\mathbf{m}_\sigma^r(i) \ne \mathbf{m}_\sigma^r(i-1)$. 
Hence for each locally maximal alternating substring 
$\mathbf{m}_\sigma^r(a, a+k-1)$, then 
$|D(\mathbf{m}_\sigma^r)|$ is decreased by at least 
$(k-2)(n-2 - (r-1)) \ge (k-2)(r-1)$ when compared to 
$|D(\mathbf{m}_e^r)|$. 
Meanwhile, $|E(\mathbf{m}_\sigma^r)|$ 
is increased by the same 
locally maximal alternating substring by at most 
$(k-2)((k-2)/4)$ by Lemma \ref{oddeven}. 
However, since $k \le 2r$, we have 
$(k-2)((k-2)/4) \le (k-2)(r-1)/2$, which is 
of course less than $(k-2)(r-1)$. 
\end{proof}

Lemma \ref{esphere}, along with 
Proposition \ref{onesphere} implies that 
the $r$-regular Ulam sphere size of radius $t=1$ 
is bounded (tightly) below by $(1 + (n-1)(n/r-1))$. 
This in turn implies the 
following sphere-packing type upper bound on 
any single error-correcting code.

\begin{lemma}\label{upperbound}
Let $C$ be a single-error correcting
 $\mathsf{MPC}_\circ(n,r)$ code. 
Then 
\[
|C|_r \le \frac{n!}{(r!)^{n/r}\left(1+(n-1)(n/r-1)\right)}.
\]
\end{lemma}

\begin{proof}
Let $C$ be a single-error correcting 
$\mathsf{MPC}_\circ(n,r)$ code. 
A standard sphere-packing bound argument 
implies that 
$|C|_r \le \frac{n!}{(r!)^{n/r}
(\underset{\sigma \in \mathbb{S}_n}{\min} 
\left(
|S(\mathbf{m}_\sigma,1)|
\right)}$.
The remainder of the proof follows from 
Proposition \ref{onesphere} and Lemma \ref{esphere}.
\end{proof}

We have seen that $|S(\mathbf{m}_\sigma^r)|$ is minimized 
when $\sigma = e$. We now discuss the choice of center 
yielding the maximal sphere size. 
Let 
$\omega \in \mathbb{S}_n$ be defined as follows: 
$\omega(i) := ((i-1)\mod (n/r))r + \lceil ir/n \rceil $ 
and $\omega := [\omega(1), \omega^*(2), \dots \omega^*(n)].$ 
With this definition, for all $i \in [n],$ we have 
$\mathbf{m}_{\omega}^r(i) = i \mod (n/r)$ 
For example, if $r = 3$ and $n=12,$ then 
$\omega = [1, 4, 7, 10, 2, 5, 8, 11, 3, 6, 9, 12]$ and 
$\mathbf{m}_{\omega}^r = (1, 2, 3, 4, 1, 2, 3, 4, 1, 2, 3, 4).$ 
We can use Theorem \ref{ballcalc} to calculate 
$|S(\mathbf{m}_\omega^r,1)|$, and then 
show that this is the 
maximal $r$-regular Ulam sphere size 
(except for the case when $n/r = 2$).

\begin{lemma}\label{omegasphere}
Let 
$n/r \ne 2.$  
Then 
$|S(\mathbf{m}_\sigma^r,1)| 
\le |S(\mathbf{m}_{\omega}^r,1)|$ 
and if $n/r > 2,$ then 
$|S(\mathbf{m}_{\omega}^r,1)| 
=  (1+(n-1)^2) - (r-1)n$.
\end{lemma}

\begin{proof}
Assume $n/r \ne 2.$  
First notice that if $n/r = 1$ then for any $\pi \in \mathbb{S}_n$ 
(including $\pi = \omega$), the sphere
$S(\mathbf{m}_{\pi}^r,1)$
contains exactly one 
element (the tuple of the form $(1,1,\dots,1)$). 
Hence the lemma holds trivially in this instance. 
Next, assume that $n/r > 2.$  
We will first prove that $|S(\mathbf{m}_{\omega}^r,1)| 
=  (1+(n-1)^2) - (r-1)n$.

Since $n/r > 2,$ it is clear that 
 $\mathbf{m}_{\omega}^r$ contains no 
alternating subsequences of length greater than $2$. 
Thus by Lemma \ref{remark2}, 
$E(\mathbf{m}_{\omega}^r) = \varnothing$ 
and therefore by Theorem \ref{ballcalc}, 
$|S(\mathbf{m}_{\omega}^r,1)| = 
1 + (n-1)^2 - |D(\mathbf{m}_{\omega}^r)|.$ 
Since there does not exist $i \in [n]$ such that 
$\mathbf{m}_{\omega}^r (i) = \mathbf{m}_{\omega}^r(i-1)$, 
we have $|D(\mathbf{m}_{\omega}^r)| = (r-1)n$, 
completing the proof of the first statement in the lemma. 

We now prove that 
$|S(\mathbf{m}_\sigma^r,1)| 
\le |S(\mathbf{m}_{\omega}^r,1)|$.
Recall that 
$|D(\mathbf{m}_\sigma^r)|$ is equal to 
$(n-2)$ times the number of 
$i \in [n]$ such that 
$\mathbf{m}_\sigma^r(i) = \mathbf{m}_\sigma^r(i-1)$ plus 
$(r-1)$ times the number of $i \in [n]$ such that 
$\mathbf{m}_\sigma^r(i) \ne \mathbf{m}_\sigma^r(i-1)$. 
But $n/r > 2$ implies that $r-1 < n-2$, which implies 
$\underset{\pi \in \mathbb{S}_n}{\min}
|D(\mathbf{m}_\pi^r,1)| = (r-1)n$. Therefore 

\begin{eqnarray*}
|S(\mathbf{m}_\sigma^r,1)| &\le& 
 1 + (n-1)^2 - \underset{\pi \in \mathbb{S}_n}{\min} 
 |D(\mathbf{m}_\pi^r,1)| \\
 &&\hspace{1.91cm}  - \underset{\pi \in \mathbb{S}_n}{\min}
 |E(\mathbf{m}_\pi^r,1)| \\
 &\le& 1 + (n-1)^2 - \underset{\pi \in \mathbb{S}_n}{\min}
  |D(\mathbf{m}_\pi^r,1)|  \\
 &=& 1 + (n-1)^2 - (r-1)n \\
 &=& |S(\mathbf{m}_{\omega}^r,1)|.
\end{eqnarray*}
\normalsize

\end{proof}

Extending the concept of perfect permutation codes 
discussed in \cite{self}, we define a 
perfect multipermutation code. 
Let $C$ be an $\mathsf{MPC}(n,r)$ code. Then $C$ is a perfect $t$-error 
correcting code if and only if for all 
$\sigma \in \mathbb{S}_n,$ there exists a unique 
$\mathbf{m}_c^r \in \mathcal{M}_r (C)$ 
such that $\mathbf{m}_\sigma^r \in S(\mathbf{m}_c^r, t).$ 
We call such $C$ a \textbf{perfect $t$-error correcting } 
$\mathsf{MPC}(n,r)$. 
With this definition the upper bound of lemma \ref{omegasphere}
implies a lower bound on a perfect single-error correcting 
$\mathsf{MPC}(n,r)$. 

\begin{lemma}\label{perfect bound}
Let $n/r \ne 2$ and let $C$ be a perfect single-error
 correcting $\mathsf{MPC}(n,r)$.Then 
\[
|C|_r \ge \frac{n!}{(r!)^{n/r}  ((1+(n-1)^2) - (r-1)n)}.
\]
\end{lemma}

\begin{proof}
Suppose $n/r \ne 2$ and $C$ is a perfect single-error 
correcting $\mathsf{MPC}(n,r)$. 
Then $\underset{c \in C}{\sum}
|S(\mathbf{m}_c^r,1)| = \frac{n!}{(r!)^{n/r}}$. 
This means 
\begin{eqnarray*}
\left(|C|_r \right) \cdot
\left( \underset{c \in C}{\max}
\left(|S(\mathbf{m}_c^r,1)|\right)\right)
\ge \frac{n!}{(r!)^{n/r}},
\end{eqnarray*} 
which by Lemma \ref{omegasphere} 
implies the desired result. 
\end{proof}


A more general lower bound is easily
obtained by applying Lemma \ref{omegasphere} with a 
standard Gilbert-Varshamov bound argument.
In the lemma statement, $C$ is an $\mathsf{MPC}_\circ(n,r,d)$
 if and only if $C$ is an $\mathsf{MPC}(n,r)$ such that 
$\underset{\sigma,\pi \in C, \sigma \ne \pi}{\min} 
\mathrm{d}_\circ^r(\sigma,\pi)=d.$  

\begin{lemma}\label{G-V bound}
Let $n/r \ne 2$ and $C$ be an 
$\mathsf{MPC}_\circ(n,r,d)$ code of maximal cardinality. Then 
\[
|C|_r \ge \frac{n!}{(r!)^{n/r} (1 + (n-1)^2 - (r-1)n )^{d-1}  }
\]
\end{lemma}

\begin{proof}
Suppose that $n/r \ne 2$ and that 
$C$ is an $\mathsf{MPC}_\circ(n,r,d)$ code
of maximal cardinality. 
For all $\sigma \in \mathbb{S}_n$, 
there exists 
$c \in 
C$ such that 
$\mathrm{d}_\circ^r(\sigma,c) \le d-1$. 
Otherwise, we could add 
$\sigma \notin C$ (and its entire equivalence class 
$R_r(\sigma)$) to $C$ while maintaining 
a minimum distance of $d$, contradicting 
the assumption that $|C|_r$ is maximal. 

Therefore 
$\underset{c \in C}{\bigcup}S(\mathbf{m}_c^r,d-1)
= \mathcal{M}_r(\mathbf{S}_n)$.
This in turn implies that 
\begin{eqnarray*}
\underset{c\in C}{\sum}|S(\mathbf{m}_c^r,d-1)| 
\ge \frac{n!}{(r!)^{n/r}}. 
\end{eqnarray*}
Of course, the left hand side of the 
above inequality is less than or equal to 
$\left(|C|_r\right) \cdot
\left( \underset{c \in C}{\max}
\left(|S(\mathbf{m}_c^r,d-1)|\right)\right)$. 
Hence Lemma \ref{omegasphere} implies that 
\[(1+(n-1)^2 - (r-1)n)^{d-1} \ge 
\underset{c\in C}{\max}(|S(\mathbf{m}_c^r,d-1)|),\] 
so the 
conclusion holds. 
\end{proof}

\section{Conclusion} \label{conclusion}
This paper compared the Ulam metric for the permutation 
and multipermutation cases, providing a simplification of 
the $r$-regular Ulam metric. The surprising fact that 
$r$-regular Ulam sphere sizes differ depending upon the 
center was also shown. New methods for calculating the 
size of $r$-regular Ulam sphere sizes were provided, first
using Young Tableaux for spheres of any radius centered at 
$\mathbf{m}_e^r$.  Another method  used duplicate 
translocation sets to calculate sphere sizes for  a radius 
of $t=1$ for any center. Resulting bounds on Code size 
were also provided. Many open questions remain, including 
the existence of perfect codes, sphere size calculation
methods for more general parameters, and tighter 
bounds on code size.

\section*{Acknowledgment}
This paper is partially supported by
KAKENHI 16K12391 and 26289116.

\end{document}